\documentclass[12pt,english]{article}

\usepackage{titlesec}
\usepackage{lipsum}% just to generate text for the example
\titlespacing*{\section}
{0pt}{1.7ex plus 1ex minus .2ex}{1.2ex plus .2ex}
\titlespacing*{\subsection}
{0pt}{1.5ex plus 1ex minus .2ex}{1.0ex plus .2ex}

\usepackage{fullpage}
\usepackage{amsmath, amssymb, amsthm}
\usepackage{enumerate}
\usepackage{fullpage}
\usepackage{xr}
\usepackage{bbold}
\usepackage{mathtools}
\usepackage[export]{adjustbox}
\usepackage{graphicx}
\usepackage{setspace}
\usepackage[short]{optidef}
\usepackage{color}
\usepackage{float}
\usepackage{pst-all}
\usepackage{natbib}
\usepackage{etoolbox}
\usepackage{comment}
\usepackage{framed}
\usepackage{xcolor}
\usepackage{tikz}
\newcommand{\lightercolor}[3]{% Reference Color, Percentage, New Color Name
    \colorlet{#3}{#1!#2!white}
}
\lightercolor{gray}{40}{shadecolor}
 \definecolor{lightgray}{gray}{0.85}
 \definecolor{mygray}{gray}{0.8}
 \definecolor{mylightblue}{gray}{0.6}

\usepackage[T1]{fontenc}
\usepackage{babel}
\usepackage{diagbox}
\usepackage[bottom]{footmisc}
\usepackage{breakcites}
\usepackage[normalem]{ulem}

\usepackage[breaklinks,hidelinks]{hyperref}
\usepackage{url}
\usepackage{breakurl}

\newlength{\bibitemsep}
\newlength{\bibparskip}
\let\oldthebibliography\thebibliography
\renewcommand\thebibliography[1]{%
  \oldthebibliography{#1}%
  \setlength{\parskip}{\bibitemsep}%
  \setlength{\itemsep}{\bibparskip}%
}

\newcommand{\cale}{\mathcal E}

\newcommand{\underv}{\underline{v}}
\newcommand{\under}{\underline}

\DeclareMathOperator*{\argmax}{argmax}

\newcommand{\support}{\text{support}}
\newcommand{\overU}{{\overline{U}}}
\newcommand{\underU}{{\underline{U}}}

\newcommand{\underu}{{\underline{u}}}

\newcommand{\bfr}{\mathbf{R}}

\theoremstyle{remark}
\newtheorem{remark}{Remark}
\newtheorem{case}{Case}

\newenvironment{proofof}[1]{\begin{proof}[Proof of #1]}{\end{proof}}
\usepackage{accents}
\theoremstyle{plain}
\newtheorem{theorem}{Theorem}[section]
\newtheorem{lemma}[theorem]{Lemma}
\newtheorem{proposition}{Proposition}[section]

\newtheorem{claim}[theorem]{Claim}

\theoremstyle{definition}

\newcommand{\rgrt}{\mathrm{RGRT}}

\newcommand{\wcr}{\mathrm{WCR}}
\title{Regret-Minimizing Project Choice}
\author{Yingni Guo \and Eran Shmaya\thanks{Guo: Department of Economics, Northwestern University; email: \href{mailto:yingni.guo@northwestern.edu}{yingni.guo@northwestern.edu}. Shmaya: Department of Economics, Stony Brook University; email: \href{mailto:eran.shmaya@stonybrook.edu}{eran.shmaya@stonybrook.edu}. We thank seminar audiences at the One World Mathematical Game Theory Seminar, the Toulouse School of Economics, the University of Bonn, University of Pennsylvania, Brown University, Northwestern University Theory Lunch Seminar, Kellogg Strategy Department Lunch Seminar, the University of Pittsburgh, Carnegie Mellon University, Stony Brook International Conference on Game Theory, Workshop on Strategic Communication and Learning, TE Session at ASSA, Arizona State University, Stanford University, Cornell University, Paris Dauphine University, MIT/Harvard, NYU, and Erice Workshop on Stochastic Methods in Game Theory for valuable feedback.}}

\date{\today}

\begin{document}

\maketitle

\begin{abstract} 

An agent observes the set of available projects and proposes some, but not necessarily all, of them. A principal chooses one or none from the proposed set. We solve for a mechanism that minimizes the principal's worst-case regret. We compare the single-project environment in which the agent can propose only one project with the multiproject environment in which he can propose many. In both environments, if the agent proposes one project, it is chosen for sure if the principal's payoff is sufficiently high; otherwise, the probability that it is chosen decreases in the agent's payoff. In the multiproject environment, the agent's payoff from proposing multiple projects equals his maximal payoff from proposing each project alone. The multiproject environment outperforms the single-project one by providing better fallback options than rejection and by delivering this payoff to the agent more efficiently.


\end{abstract}

\section{Introduction}

Project choice is an important function of many organizations.\footnote{For instance, \cite{Ross1986} surveyed twelve large manufacturers in 1981-82. In the area of energy preservation alone, he collected data on $400$ discretionary projects. Roughly $100$ projects were completed in 1980-81, and $300$ projects were either underway or prospective, intended for the period 1982-85.} The process often involves two parties: (i) a party at a lower hierarchical level who has expertise and proposes projects, and (ii) a party at a higher hierarchical level who evaluates the proposed projects and makes the choice. This describes, for example, the relationship between a firm's division and the firm's headquarters when the division has a chance to choose an investment project, a factory location, or an office building.\footnote{\cite{StanleyBlock1984} surveyed $121$ multinational firms; $88$ percent of the respondents reported that projects were initiated from bottom up, and $70$ percent reported that decision-making was centralized.} It also applies to the relationship between a department and the university administration when the department has a hiring slot open.

This process of project choice is naturally a principal-agent problem. The agent privately observes which projects are available and proposes a subset of the available projects. The principal either chooses one from the projects proposed or rejects them all. If the two parties had identical preferences over projects, the agent would propose their jointly favorite project from the available ones, and the principal would simply rubber-stamp the agent's proposal. In reality, however, the principal does exert power over project choice, because the agent's preferences are not necessarily aligned with the principal's. For example, division managers may fail to internalize projects' externalities on other divisions; they may prefer projects that require large investments due to empire-building incentives (\cite{HarrisRaviv1996}, \cite{BerkovitchIsrael2004}); or the department and the university may put different weights on candidates' research and nonresearch attributes. Armed with the proposal-setting power, the agent has a tendency to propose his favorite project and hide his less preferred ones, even if those projects are ``superstars'' for the principal. 

This paper explores how the principal counteracts the agent's proposing bias. We address this question using a robust-design, non-Bayesian approach. Such an approach is useful when the principal faces a certain project-choice problem for the first time. It is also applicable to organizations in which the principal oversees many project-choice problems, but each problem offers little guidance on how to formulate a prior belief for another. For instance, the headquarters might oversee project choice in various divisions which operate in quite different markets. Similarly, the university oversees hiring in a range of departments; moreover, the academic labor market varies widely not only across departments but also across years. 

Due to the agent's private information, no mechanism can guarantee that the principal's favorite project among the available ones will be chosen. We define the principal's \emph{regret} as the difference between his payoff from his favorite project and his expected payoff from the project chosen under the mechanism. Thus, regret can be interpreted as ``money left on the table'' due to the agent's private information. We characterize a mechanism that minimizes the principal's worst-case regret, i.e., a mechanism that leaves as little money on the table as possible in all circumstances. As we will show, this minimax regret approach allows us (i) to explain why certain incentive schemes are common, (ii) to propose new incentive schemes, and (iii) to explore questions which are intractable under the Bayesian approach. 

\vspace{-5 mm}
\paragraph{Main results.} The principal chooses at most one project. Yet, one often observes that the agent proposes multiple projects.\footnote{For instance, it is usual practice for a search committee to propose multiple finalists for an appointment as a dean at a business school (\cite{Byrne2018,Byrne2021}).\label{ft:multi}} The strategic role that a multiproject proposal plays in project-choice problems has not yet been explored. Do several projects within the proposal have a chance of being chosen? How does the principal benefit from allowing the agent to propose multiple projects rather than to propose only a single project?  

To answer these questions, we distinguish between two environments. In the \emph{multiproject} environment, the agent can propose any subset of the available projects. In the \emph{single-project} environment, the agent can propose only one available project. Such a single-project restriction may emerge in organizations due to the principal's limited attention. It is also applicable to antitrust regulation, where a firm proposes one merger and the regulator decides whether to approve or reject the firm's proposal (\cite{Lyons2003}, \cite{Neven2005}, \cite{ArmstrongVickers2010}, \cite{Ottaviani2011}, \cite{NockeWhinston}). We take the environment as exogenous and derive the optimal mechanism in each environment.

In the single-project environment, a mechanism specifies for each single-project proposal the probability that it will be approved. We show that the optimal mechanism has a two-tier structure. If the proposed project gives the principal a sufficiently high payoff, it is a top-tier project for the principal and is approved for sure. Otherwise, it is a bottom-tier project and is approved only with some probability. The probability that a bottom-tier project is approved decreases in its payoff to the agent, in order to deter the agent from hiding projects that are more valuable for the principal. 

The optimal two-tier mechanism in the single-project environment aligns the agent's proposing incentives with the principal's preferences in two ways. First, if the agent has at least one top-tier project, then he will propose a top-tier project. Second, if all his projects are bottom-tier ones, then he will propose a project that maximizes the principal's expected payoff. For the principal to reject a bottom-tier project is suboptimal ex post, but is indispensable for aligning the agent's proposing incentives ex ante.

In the multiproject environment, a mechanism specifies for each proposed set of projects if and how a project will be chosen, namely, a randomization over the proposed projects and ``no project.'' We first show that the revelation principle holds in this environment, so it is without loss of generality to focus on mechanisms under which the agent optimally proposes all the available projects. 

If the agent in the multiproject environment actually proposes a single project, the optimal mechanism again has a two-tier structure. In particular, if this project's payoff to the principal is sufficiently high, he views it as a top-tier project and chooses it for sure. Otherwise, the project is a bottom-tier one and is chosen with some probability that decreases in its payoff to the agent. 

If the agent proposes multiple projects, his expected payoff must be at least his maximal expected payoff from proposing each project alone. This is because the agent had the option to propose just one of the available projects. The optimal mechanism randomizes over the proposed projects and ``no project'' to maximize the principal's expected payoff, subject to the constraint of promising this maximal expected payoff to the agent.\footnote{This maximization problem is a linear programming problem. There exists an optimal solution whose support has at most two projects.} Since the agent gets his maximal expected payoff from proposing each project alone in the proposal, we call this mechanism a \emph{proposal-wide maximal-payoff mechanism} (PMP mechanism). 

The optimal PMP mechanism identifies the strategic role that multiproject proposals play in project-choice problems. Among the available projects, there is one project that gives the agent his maximal expected payoff from proposing each project alone. This project is what the agent himself wants to propose and pins down his expected payoff. The multiproject environment allows the agent to also propose other projects, so it can deliver this expected payoff to the agent more efficiently than the single-project environment does.

\begin{figure}[htb!]
\begin{center}
\psset{xunit=4cm,yunit=4cm}

\begin{pspicture}(1.72,-.09)(3.3,1.17)

%\psline(1.72,-.09)(3.3,-.09)(3.3,1.14)(1.72,1.14)(1.72,-.09)

\psaxes{->,ticks=none,labels=none}(1.8,0)(1.8,0)(2.85,1.05)

\rput[l](2.86,0){agent payoff}
\rput[c](2.0,1.1){principal payoff}
\rput[c](1.8,-.06){$(0,0)$}

\rput[c](1.95,.85){$\bigstar$}
\rput[c](2.6,.5){{\large $\blacktriangle$}}

\psline[linecolor=black, linestyle=dotted, linewidth=1.2pt](1.8,0)(2.6,.5)
\psline[linecolor=black, linestyle=dotted, linewidth=1.2pt](1.95,.85)(2.6,.5)

\psline[linecolor=black, linestyle=dotted, linewidth=1.4pt](1.8,.25)(2.2,.25)
\psline[linecolor=black, linestyle=dotted, linewidth=1.4pt](2.2,0)(2.2,0.715385)
\psline[linecolor=black, linestyle=dotted, linewidth=1.4pt](1.8,0.715385)(2.2,0.715385)

\rput[c](2.2,-0.05){$x$}
\rput[c](1.75,.25){$y$}
\rput[c](1.75,0.715385){$z$}

\end{pspicture}
\end{center}
\vspace{-6 mm}
\caption{Strategic role of multiproject proposals}
\label{fg:intro}
\end{figure}
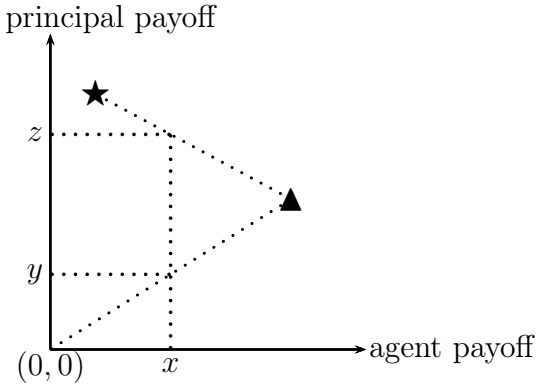

To illustrate this strategic role, consider the example in Figure \ref{fg:intro}. The agent has two projects, $\bigstar$ and $\blacktriangle$. According to the optimal mechanism in subsection \ref{se:multi}, if the agent proposed $\blacktriangle$ alone, it would be chosen with probability $1/2$. The agent's and the principal's expected payoffs would be $x$ and $y$, respectively. Since $x$ is higher than the agent's payoff from $\bigstar$, proposing $\blacktriangle$ gives the agent his maximal payoff from proposing each project alone. The multiproject environment allows the agent to also propose $\bigstar$, a better fallback option than rejection. The mechanism randomizes between $\blacktriangle$ and $\bigstar$ while promising the agent an expected payoff of $x$. This lifts the principal's expected payoff from $y$ to $z$.

Our result also explains when more than one proposed project have a chance of being chosen. If the agent's favorite project is already a top-tier project for the principal, the two parties' preferences are somewhat aligned. We show that this project will be chosen for sure. In contrast, if the agent's favorite project is a bottom-tier project for the principal, their preferences are less aligned. The mechanism may randomize between multiple proposed projects. Figure \ref{fg:intro} gives such an example: the mechanism ``finds some middle ground'' between the two parties---with some probability the choice favors the agent and with some probability it favors the principal.

We further analyze the environment in which the agent can propose up to $K$ projects for some fixed $K \geqslant 2$. We show that the aforementioned PMP mechanism admits an implementation in which (i) the agent proposes up to two projects, and (ii) the principal achieves the same worst-case regret as he does in the multiproject environment. The first project is what gives the agent his maximal expected payoff from proposing each project alone; this is what the agent himself wants to propose. The second is the principal's favorite project; this can be the fallback option when the first project is not chosen. Hence, in terms of minimizing the principal's worst-case regret, allowing for up to two projects is enough. 

Our results offer a number of insights and implications for the practical design of project-choice rules.

(1) The optimal mechanism in the single-project environment is comparable to how projects are chosen in the firms surveyed by \cite{Ross1986}. We discuss this connection in detail in subsection \ref{se:committorandomization}.  In the single-project environment, the threshold for being a top-tier project is higher than a zero payoff for the principal; such a project is approved for sure and the agent never proposes a bottom-tier project when he has a top-tier one. This high threshold for top-tier projects suggests one possible explanation for the finding in \cite{PoterbaSummers1995}. They show that US firms, when evaluating investment projects, use hurdle rates which are considerably higher than the cost of capital. In academic departments, top-tier projects are sometimes called slam-dunk cases or star candidates. In our experience, it is a no-brainer to propose such a candidate, who is typically approved with no questions asked. In contrast to the sure approval of top-tier projects, bottom-tier ones are rejected with positive probabilities. This result provides one possible explanation for why, in the case of faculty hiring, no slot is ever guaranteed and random rejection occurs. 

(2) We show that the principal's worst-case regret in the multiproject environment is significantly lower than that in the single-project one. Our result suggests one possible explanation for the practice of proposing more projects than are needed to the principal. In particular, senior positions like a deanship are vital to the success of an organization. It is reasonable to expect that for such positions the benefit from the multiproject environment outweighs the cost of reviewing multiple candidates, so the multiproject environment is more likely to be adopted for such positions than for lower-level ones. We show that in the multiproject environment, the principal does not always choose his own favorite one among those proposed. One testable implication is that, if an outsider collects all the interactions where multiple projects are proposed to the principal, the choice is not always the principal's favorite one, but sometimes reflects the proposing party's preferences.

\vspace{-5 mm}

\paragraph{Related literature.}

Our paper is closely related to \cite{ArmstrongVickers2010} and \cite{NockeWhinston}, who study the project-choice problem using the Bayesian approach. \cite{ArmstrongVickers2010} characterize the optimal deterministic mechanism in the single-project environment and show through examples that the principal does strictly better if either randomization or multiproject proposals are allowed. Their discussion highlights the informational role of multiproject proposals.\footnote{The informational role of multiproject proposals can be explained by the following example. If the principal is quite certain that there are three available projects, he can get the agent's information almost for free in the multiproject environment. The mechanism chooses the principal's favorite project if three projects are proposed and chooses no project otherwise.} We complement their discussion by identifying the strategic role of multiproject proposals. \cite{NockeWhinston} focus on mergers (i.e., projects) that are ex ante different and further incorporate the bargaining process among firms. They show that a tougher standard is imposed on mergers involving larger partners. We depart from these papers by taking the minimax regret approach to this multidimensional screening problem. This more tractable approach allows us to explore questions which are intractable under the Bayesian approach, including how much the principal benefits from the multiproject environment, from randomization, and from a smaller project domain. 

If the principal is restricted to the single-project environment and deterministic mechanisms, a mechanism is effectively a set of permitted projects that the agent can choose with no questions asked. This connects project-choice problems to the delegation literature initiated by \cite{Holmstrom1984}.\footnote{Some later contributions include, for example, \cite{MelumadShibano1991}, \cite{MartimortSemenov2006}, \cite{AlonsoMatouschek2008}, and \cite{AmadorBagwell2013}.} If either the multiproject environment or randomization is allowed, more mechanisms can be implemented than by simply specifying a set of permitted projects (as shown, respectively, in subsection 3.3 and footnote 8 of \cite{ArmstrongVickers2010}). Our contribution is to characterize the optimal randomized mechanisms in both the single-project and the multiproject environments.

\cite{BerkovitchIsrael2004} study a project-choice model with binary project types (high versus low capital-intensive). The agent has an incentive to hide the principal's preferred project (which is low capital-intensive). They solve for the optimal randomized mechanism in which the principal can compensate the agent for certain proposals. They characterize conditions under which monetary incentives are not used. The principal rejects the high capital-intensive project with a positive probability, so the agent always reveals the low capital-intensive project whenever it is available. Our two-tier mechanism in the single-project environment has a similar feature: bottom-tier projects are rejected with positive probabilities so the agent proposes a top-tier project if he has at least one such project. 
 
\cite{CheDesseinKartik2013} and \cite{GoelHannCaruthers} consider a project-choice problem in which the number of available projects is public information. In \cite{CheDesseinKartik2013}, the agent privately learns the projects' payoffs and sends a cheap-talk message. They show that the principal benefits from randomizing between a default project and the agent's recommendation. In \cite{GoelHannCaruthers}, the agent's only constraint is not to overreport projects' payoffs to the principal. Because the agent in these papers cannot hide projects as our agent can, he loses the proposal-setting power. The resulting incentive schemes are thus quite different.
 
The minimax regret approach to uncertainty dates back to \cite{Wald1950} and \cite{Savage1951}. It has since been used broadly in game theory, mechanism design, and machine learning. A decision-theoretical axiomatization for the minimax regret criterion can be found in \cite{milnor1954games} and \cite{Stoye2011}. Our paper contributes especially to the literature on mechanism design with the minimax regret approach. To start with, \cite{HurwiczShapiro1978} examine a moral hazard problem. \cite{BergemannSchlag2008, BergemannSchlag2011} examine monopoly pricing. \cite{RenouSchlag2011} apply the solution concept of $\varepsilon$-minimax-regret to the problem of implementing social choice correspondences. \cite{Bevia2019} examine the contest which minimizes the designer's worst-case regret. \cite{malladi2020judged} studies the optimal approval rules for innovation, and \cite{GuoShmaya2023} study the optimal mechanism for monopoly regulation. More broadly, we contribute to the growing literature of mechanism design with worst-case objectives. For a survey on robustness in mechanism design, see \cite{Carroll2019}.

\section{Model and mechanism}
%Let $u:D\to \bfr_+$ be the agent's payoff function, so his payoff is $u(a)$ if project $a$ is chosen. If no project is chosen, the agent's payoff is zero. Let $v:D\to \bfr_+$ be the principal's payoff function, so his payoff is $v(a)$ if project $a$ is chosen. If no project is chosen, the principal's payoff is zero.

Let $D$ be the domain of all possible projects. The agent's private \emph{type} $A \subseteq D$ is a finite set of available projects. The agent proposes a set $P$ of projects, and the principal can choose one project from this set. The set $P$ is called the agent's \emph{proposal}. It must satisfy two conditions. First, the agent's proposal is \textit{verifiable} in the sense that he can propose only available projects. Hence, the agent's proposal must be a subset of his type, $P\subseteq A$. %(\cite{GreenLaffont1986} model verifiable information by assuming that the agent's message space varies with his type. In our model, the set of the proposals the agent can make varies with his type.) 
Second, $P\in \cale $ for some fixed set $\cale$ of subsets of $D$. The set $\cale$ captures all the exogenous restrictions on the proposal. For instance, in the setting of antitrust regulation, the agent is restricted to proposing at most one project. In some organizations, the principal has limited attention, so the agent can propose at most a certain number of projects.  

Let $u:D\to \bfr_+$ be the agent's payoff function and $v:D\to \bfr_+$ the principal's payoff function. If project $a$ is chosen, the agent's payoff is $u(a)$ and the principal's payoff is $v(a)$. If no project is chosen, both parties get zero. 

We begin with two environments which are natural first steps: \emph{single-project} and \emph{multiproject}. In the single-project environment, the agent can propose at most one available project, so $\cale=\{P\subseteq D: |P|\leqslant 1 \}$. In the multiproject environment, the agent can propose any set of available projects so $\cale=2^D$, the power set of $D$. In subsection \ref{se:intermediate}, we discuss the intermediate environments in which the agent can propose up to $K$ projects for some fixed number $K\geqslant 2$.

A \emph{subprobability measure over $D$ with a finite support} is given by $\pi:D \to [0,1]$ such that \[\support(\pi)=\{a\in D:\pi(a)> 0\}\] is finite, and $\sum_{a}\pi(a)\leqslant 1$. When we say that a project \emph{is chosen from} a subprobability measure $\pi$ with finite support, we mean that project $a$ is chosen with probability $\pi(a)$, and that with probability $1-\sum_{a}\pi(a)$ no project is chosen. 

The principal's ability to reject all proposed projects (i.e., to choose no project) is crucial for him to retain some bargaining power. If the principal were required to choose a project as long as the agent has proposed some, then the agent would propose only his own favorite project, which would be chosen for sure.

A \emph{mechanism} $\rho$ attaches to each proposal $P\in \cale$ a subprobability measure $\rho(\cdot|P)$ such that $\support(\rho(\cdot|P))\subseteq P$. The interpretation is that, if the agent proposes $P$, then a project is chosen from the subprobability measure $\rho(\cdot|P)$. Thus, the agent's expected payoff under the mechanism $\rho$ if he proposes $P$ is $U(\rho,P)=\sum_{a\in P} u(a)\rho(a|P)$.

A \emph{choice function} $f$ attaches to each type $A$ of the agent a subprobability measure $f(\cdot|A)$ such that $\support(f(\cdot|A))\subseteq A$. The interpretation is that, if the set of available projects is $A$, then a project is chosen from the subprobability measure $f(\cdot|A)$. 

A choice function $f$ is \emph{implemented} by a mechanism $\rho$ if, for every type $A$ of the agent, there exists a probability measure $\mu$ with support over $\argmax_{P\subseteq A, P\in \cale} U(\rho,P)$ such that $f(a|A)=\sum_{P}\mu(P)\rho(a|P)$ for every $a\in A$. The agent chooses only proposals that give him the highest expected payoff and can randomize among them. 

%The interpretation is that the agent can randomize among the proposals that give him the highest expected payoff. 

% The interpretation is that the agent randomizes the proposal among all the proposals that he can make that give him the highest expected payoff. 

\subsection{Revelation principle in the multiproject environment}

Consider the multiproject environment $\cale=2^D$. A mechanism $\rho$ is \emph{incentive-compatible (IC)} if the agent finds it optimal to propose all the available projects. That is, $U(\rho,A)\geqslant U(\rho,P)$ for every finite set $A\subseteq D$ and every subset $P \subseteq A$. Equivalently, a mechanism $\rho$ is IC if and only if the agent's payoff, $U(\rho,P)$, weakly increases in $P$ with respect to set inclusion. The following proposition states the revelation principle in the multiproject environment. 
\begin{proposition}\label{pr:revelation}
  Assume $\cale=2^D$. If a choice function $f$ is implemented by some mechanism, then the mechanism $f$ is IC and implements the choice function $f$.\end{proposition}
  
 When the agent proposes a set $P$ of projects, he provides evidence that his type $A$ satisfies $P\subseteq A$. In the multiproject environment, the agent can provide the maximal evidence for his type. This property is called \emph{normality} by \cite{BullWatson2007}, which is equivalent to the \emph{full reports condition} of \cite{LipmanSeppi1995}. Another interpretation of the multiproject environment is to view an agent who proposes a set $P$ as an agent who declares that his type is $P$. The relation that ``type $A$ can declare to be type $B$'' between types is reflexive and transitive, by the corresponding properties of the inclusion relation between sets. Transitivity is called the \emph{nested-range condition} in \cite{GreenLaffont1986}. Previous papers (e.g., \cite{GreenLaffont1986}, \cite{BullWatson2007}) have shown that the revelation principle holds under these conditions. We repeat the argument for completeness. %Although our Proposition~\ref{pr:revelation} does not follow directly from their theorems (because the agent's proposal $P$ not only provides evidence but also determines the set of projects from which the principal can choose), our proof uses a very similar argument as in those papers.  

\begin{proofof}{Proposition~\ref{pr:revelation}}
Assume that the mechanism $\rho$ implements the choice function $f$. Then for every finite set $A \subseteq D$ and every subset $P\subseteq A$, we have:
\[U(f,A)=\max_{Q \subseteq A} U(\rho,Q)\geqslant  \max_{Q \subseteq P} U(\rho,Q)=U(f,P),\]
where the inequality follows from the fact that $Q\subseteq P$ implies  $Q\subseteq A$, and the two equalities follow from the fact that $\rho$ implements $f$. Hence, the mechanism $f$ is IC. Also, by definition, if the mechanism $f$ is IC, then it implements the choice function $f$.
\end{proofof}

Since an implementable choice function is itself an IC mechanism and vice versa, we will use both terms interchangeably when we discuss the multiproject environment.

\section{The principal's problem}

The principal's \emph{regret} from a choice function $f$ when the set of available projects is $A$ is:
\[
\rgrt(f,A)=\max_{a \in A}v(a)-\sum_{a\in A}v(a)f(a|A).
\]
The regret is the difference between what the principal could have achieved if he had known the set $A$ of available projects and what he actually achieves. The \emph{worst-case regret} ($\wcr$) from a choice function $f$ is:
\[
\wcr(f)=\sup_{A\subseteq D, |A|<\infty  }\rgrt(f,A),
\]
where the supremum ranges over all possible types of the agent (i.e., all possible finite sets of available projects). The principal's problem is to minimize $\wcr(f)$ over all implementable choice functions $f$. This step is our only departure from the Bayesian approach. The Bayesian approach instead assigns a prior belief over the number and the payoff characteristics of the available projects. The principal's problem under the Bayesian approach is to minimize the \emph{expected} regret instead of the \emph{worst-case} regret.

%\cite{Savage1951} calls this difference \emph{loss}. We call it \emph{regret} instead, and thereby follow the more recent game theory and computer science literature. 

\cite{Wald1950} and \cite{Savage1972} propose considering only \emph{admissible} choice functions (i.e., choice functions that are not weakly dominated). An implementable choice function $f$ is \emph{admissible} if there exists no other implementable $f'$ such that the principal's regret is weakly higher under $f$ than under $f'$ for every type of the agent and strictly higher for some type. For the rest of the paper, we focus mainly on admissible choice functions. 

While our principal takes the minimax regret approach towards the agent's type, he calculates the expected payoff with respect to his own objective randomization. The same assumption is made by \cite{Savage1972} when he discusses the use of randomized acts under the minimax regret approach \cite[chapter~9.3]{Savage1972}. A similar assumption is made in the ambiguity-aversion literature. For example, in~\cite{Gilboa1989}, the decision maker calculates his expected payoff with respect to random outcomes (i.e., ``roulette lotteries'') but evaluates acts using the maxmin approach with respect to non-unique priors. %If we made the alternative assumption that the principal takes the worst-case regret approach even towards his own randomization, we effectively restrict the principal to deterministic mechanisms. 
The alternative assumption that the principal takes the minimax regret approach even towards his own randomization effectively restricts the principal to deterministic mechanisms.

From now on, we assume that the set $D$ of all possible projects is $[\under u,1]\times [\under v,1]$ for some parameters $\under u, \under v \in [0,1]$, and that the functions $u(\cdot)$ and $v(\cdot)$ are projections over the first and second coordinates. Abusing notation, we denote a project $a\in D$ also by $a=(u,v)$, where $u$ and $v$ are the agent's and the principal's payoffs, respectively, if project $a$ is chosen. 

The parameters $\underu$ and $\underv$ quantify the uncertainty faced by the principal: the higher they are, the smaller the uncertainty. They also measure the players' preference intensity over projects. As $\underu$ increases, the agent's preferences over projects become less strong, so it becomes easier to align his incentives with those of the principal. As $\underv$ increases, the principal's preferences over projects become less strong, so the agent's tendency to propose his own favorite project becomes less costly for the principal.

\section{Optimal mechanisms}
\subsection{Preliminary intuition} \label{se:pre}

We now use an example to illustrate the fundamental trade-off faced by the principal: choosing the agent's preferred project with a very high probability incentivizes the agent to propose only his preferred project, but rejecting such a project may risk rejecting the only project that is available. We also use this example to illustrate the intuition behind the optimal mechanisms. We first explain how randomization helps to reduce the principal's $\wcr$ in the single-project environment. We then explain how the multiproject environment can further reduce the $\wcr$. For this illustration, we assume that $\under v = 0$, so $D=[\under u,1]\times [0,1]$.

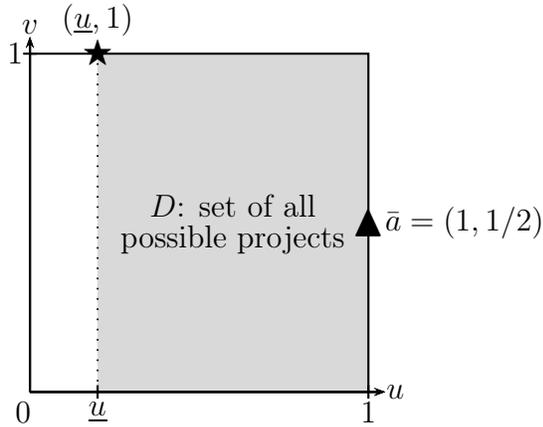
\begin{figure}[htb!]
\begin{center}
\psset{xunit=4.5cm,yunit=4.5cm}

\begin{pspicture}(-.1,-.1)(1.1,1.2)

\pscustom[linewidth=0pt,fillstyle=solid,fillcolor=lightgray,linecolor=lightgray]{ 
\psline(.2,0)(1,0)(1,1)(.2,1)}

\rput[c](.6,.55){$D$: set of all}
\rput[c](.6,.45){possible projects}

\psaxes{->,ticks=none,labels=none}(0,0)(0,0)(1.05,1.05)

\rput[c](1.08,0){$u$}
\rput[c](1,-.06){$1$}

\rput[c](0,1.08){$v$}
\rput[r](-.02,1){$1$}
\rput[c](-.02,-.06){$0$}

\psline[linecolor=black](1,-.02)(1,.02)
\psline[linecolor=black](-.02,1)(.02,1)

\psline[linecolor=black](.2,-.02)(.2,.02)
\psline[linestyle=dotted](.2,0)(.2,1)
\rput[c](.2,-.06){$\under u $}

\psline[linecolor=black,linestyle=solid](0,1)(1,1)(1,0)

\rput[c](1,.5){{\Large $\blacktriangle$}}
\rput[l](1.05,.5){$\bar{a}=(1,1/2)$}

\rput[c](.2,1){$\bigstar$}
\rput[c](.2,1.1){$(\under u, 1)$}

\end{pspicture}
\end{center}
\vspace{-6 mm}
\caption{Preliminary intuition, $\under v=0$}
\label{fg:intuition}
\end{figure}

Consider the single-project environment and assume first that the principal is restricted to deterministic mechanisms. In this case, a mechanism is a set of projects that the principal approves for sure, and all other projects are rejected outright. For each such mechanism, the principal has two fears. First, if the agent has two or more projects which will be approved, then he will propose the one he likes the most, even if projects are available that are more valuable to the principal. Second, if the agent has only projects which will be rejected, then the principal loses the payoff from these projects. As we explain in the next paragraph, these two fears imply that the principal's $\wcr$ under any deterministic mechanism is at least $1/2$.
%, no matter how he designs the deterministic mechanism. 

%In the next paragraph we argue that no matter how the principal designs the deterministic mechanism, these two fears imply that his $\wcr$ is at least $1/2$.

%Applied to the project $\bar{a}=(1,1/2)$, 

%no matter how the principal designs the deterministic mechanism, his $\wcr$ is at least $1/2$. 

Consider project $\bar{a}=(1,1/2)$ in Figure \ref{fg:intuition}. It gives the agent his highest payoff $1$, while giving the principal only a moderate payoff of $1/2$. If the mechanism approves $\bar{a}$ and the set of available projects is $\{\bar{a}, (\under u, 1)\}$, then the agent will propose $\bar{a}$ rather than $(\under u, 1)$, so the principal suffers regret $1/2$. If the mechanism rejects $\bar{a}$ but $\bar{a}$ turns out to be the only available project, then the principal also suffers regret $1/2$. Thus, the $\wcr$ under any deterministic mechanism is at least $1/2$. On the other hand, the deterministic mechanism that approves project $(u,v)$ if and only if $v\geqslant 1/2$ achieves the $\wcr$ of $1/2$, so it is optimal among all the deterministic mechanisms.

We now explain how randomization can reduce the $\wcr$ in the single-project environment. We first note that, if $\underu=0$, then, even with randomized mechanisms, the principal cannot reduce his $\wcr$ below $1/2$. This is because, when the set of available projects is $\{\bar{a},(0,1)\}$, the only way to incentivize the agent to propose the project $(0,1)$ is still to reject the project $\bar{a}$ outright if $\bar{a}$ is proposed. However, if $\underu>0$, then the principal can do better. He can approve the project $\bar{a}$ with probability $\underu$, while still maintaining the agent's incentive to propose the principal's preferred project $(\underu, 1)$ when the set of available projects is $\{\bar{a}, (\under u, 1)\}$. We carry out this idea in Theorem~\ref{th:one-project} in subsection \ref{se:single}.

Let us now consider the multiproject environment. We again begin with deterministic mechanisms. Under deterministic mechanisms, more choice functions can be implemented in the multiproject environment than in the single-project one.\footnote{For example, the principal can implement the choice function that chooses (i) the agent's favorite project, if there are at least two available projects; and (ii) nothing, if there is at most one available project.} However, when he is restricted to deterministic mechanisms, the principal has the same minimal $\wcr$ in the multiproject environment as in the single-project one. This is because, if the principal wants to choose $(\underu, 1)$ when the set of available projects is $\{\bar{a},(\underu,1)\}$, then the only way to incentivize the agent to include $(\underu,1)$ in his proposal is to reject the project $\bar{a}$ when $\bar{a}$ is proposed alone.

We now explain how randomization can help in the multiproject environment, even when $\underu=0$. While a deterministic mechanism must pick either $\bar{a}$ or $(0,1)$ or nothing when the agent proposes $\{\bar{a}, (0,1)\}$, a randomized mechanism can find some middle ground by choosing each project with probability $1/2$. On the other hand, if the agent proposes only $\bar{a}$, the principal chooses $\bar{a}$ with probability $1/2$, so the agent of type $\{\bar{a}, (0,1)\}$ is willing to propose $\{\bar{a}, (0,1)\}$ instead of just $\bar{a}$. The principal's regret is $1/4$ both when the agent's type is $\{\bar{a}, (0,1)\}$ and when his type is $\{\bar{a}\}$. We carry out this idea in Theorem~\ref{th:multiproject} in subsection \ref{se:multi}. Specifically, when the agent proposes $P$, the principal gives the agent the maximal payoff he can offer, subject to the constraint that he can give the agent this same payoff if the agent proposes $P\cup\{(\underu, 1)\}$ and can still keep his regret under control.

\subsection{Optimal mechanism in the single-project environment} \label{se:single}

Since the agent in this environment can propose at most one project, a mechanism specifies the approval probability for each proposed project. Instead of using our previous notation $\rho(a|\{a\})$, we let $\alpha(u,v) \in [0,1]$ denote the approval probability if the agent proposes the project $(u,v)$. 

 \begin{theorem}[Single-project environment]\label{th:one-project}
Assume $\cale=\{P\subseteq D: |P|\leqslant 1 \}$. Let
   \begin{equation}\label{r-one} R^s =\max_{v\in [\under v,1]}\min\left\{1-v, (1-\underu)v\right\}=\min\left\{\frac{1-\underu}{2-\underu}, 1 - \under v\right\}.
\end{equation}
\begin{enumerate}[(i)]
\vspace{-0.35cm} \item The $\wcr$ under any mechanism is at least $R^s$. 
\vspace{-0.35cm} \item Let $\alpha^s$ be the mechanism given by:
      \[\alpha^s(u,v)=\begin{cases}1, &\text{ if }v\geqslant 1-R^s \text{ or } u=0;\\ { \under u}/{u}, &\text{ if }v<1-R^s \text{ and }u>0. \end{cases}\]
Under $\alpha^s$, if the agent has projects with $v\geqslant 1-R^s$, it is optimal for him to propose his favorite project among those with $v\geqslant 1-R^s$. Otherwise, it is optimal to propose a project that maximizes the principal's expected payoff $\alpha^s(u,v)v$. The corresponding choice function has the $\wcr$ of $R^s$ and is admissible. 
\vspace{-0.35cm} \item \label{singlethree} If a mechanism $\alpha$ implements a choice function that has the $\wcr$ of $R^s$, then $\alpha(u,v)\leqslant \alpha^s(u,v)$ for every $(u,v)\in D$. 
\end{enumerate}
\end{theorem}
We first explain why the $\wcr$ under any mechanism $\alpha$ is at least $R^s$. Pick any $v\in [\underv,1]$ and consider the project $(1,v)$. If $\alpha(1,v)> \underu$, then the agent will not propose the principal's preferred project $(\underu,1)$ when both $(1,v)$ and $(\underu,1)$ are available. The principal's regret is then $1-\alpha(1,v) v \geqslant 1-v$. If $\alpha(1,v)\leqslant \underu$, then when the agent has only $(1,v)$ available, the regret is $v-\alpha(1,v)v \geqslant  (1-\underu)v$. Therefore, the principal's $\wcr$ is at least $\min\left\{1-v, (1-\underu)v\right\}$. The operator $\max_{v \in [\underv,1]}$ in $R^s$ reflects the fact that the argument holds for any $v\in [\underv,1]$.

The optimal mechanism $\alpha^s$ has a two-tier structure: a top tier if the principal's payoff $v$ is above $1-R^s$, and a bottom tier if this payoff is below $1-R^s$. A top-tier project is approved for sure. A bottom-tier project is approved with probability $\underu/u$, so the agent expects a constant payoff of $\underu$ from proposing a bottom-tier project. A top-tier project gives the principal a high enough payoff, so his regret is at most $1-v \leqslant R^s$ if he approves such a project. For a bottom-tier project $(u,v)$, the approval probability cannot exceed $\underu/u$. Otherwise, the agent will hide $(\underu,1)$ when he has both $(u,v)$ and $(\underu,1)$ available, causing high regret for the principal.

The agent will propose a top-tier project if he has at least one such project. If all his projects are in the bottom tier, he finds it optimal to propose a project that maximizes the principal's expected payoff. The principal still suffers regret from two sources. First, if the agent has two or more top-tier projects which will be approved for sure, he will propose what he favors instead of what the principal favors. Second, if the agent has only bottom-tier projects, his proposal is rejected with positive probability. The threshold for the top tier, $1-R^s$, is chosen to keep the regret from both sources under control. 

The approval probability $\alpha^s(u,v)$ increases in $v$ (the principal's payoff) and decreases in $u$ (the agent's payoff). This monotonicity in $v$ and $u$ is natural. In particular, the principal is less likely to approve projects that give the agent high payoffs in order to deter the agent from hiding projects that give the principal high payoffs. It is interesting to compare our optimal mechanism $\alpha^s$ in the single-project environment to that in \cite{ArmstrongVickers2010}. They characterize the optimal deterministic mechanism in a Bayesian setting. Under the assumptions that (i) projects are i.i.d.\ and (ii) the number of available projects is independent of their payoff characteristics, they show that the optimal deterministic mechanism $\alpha(u,v)$ increases in $v$: a project $(u,v)$ is approved for sure if $v \geqslant r (u)$ and rejected outright if $v < r (u)$. The function $r(u)$ depends on the distributions of the number and the payoffs of the available projects. Our result differs in three aspects. First, under our prior-free approach, the threshold for the top tier (i.e., where projects are approved for sure) depends only on the principal's payoff $v$. Second, we show how to reduce inefficient rejections in the bottom tier without jeopardizing the agent's incentive to propose top-tier projects. Third, we show that the approval probability $\alpha^s(u,v)$ is monotone decreasing in $u$. 

%The typical

% One possible situation under the minimax regret approach to uncertainty is that multiple mechanisms can achieve the minimal $\wcr$. 

It is possible that multiple mechanisms achieve the $\wcr$ of $R^s$. Statement (\ref{singlethree}) in Theorem \ref{th:one-project} says that the mechanism $\alpha^s$ is uniformly more generous in approving the agent's proposal than any other mechanism that has the $\wcr$ of $R^s$. This statement has two implications. First, among all mechanisms that have the $\wcr$ of $R^s$, the mechanism $\alpha^s$ is the agent's most preferred one for every possible type $A$. Second, compared to any mechanism that has the $\wcr$ of $R^s$, the mechanism $\alpha^s$ gives the principal a higher payoff (or equivalently, a lower regret) for every singleton $A$. % and a strictly higher payoff for some singleton $A$. 

\subsection{Optimal mechanism in the multiproject environment}\label{se:multi}

We now present the optimal mechanism in the multiproject environment. Let $\alpha:[\underu, 1]\times [\underv, 1]\to [0,1]$ be a function and consider the following \emph{proposal-wide maximal-payoff mechanism} (PMP mechanism) induced by the function $\alpha$: 
\begin{enumerate}
\vspace{-0.35cm} \item[(1)] If the proposal $P$ has only one project $(u,v)$, it is approved with probability $\alpha(u,v)$.
\vspace{-0.35cm} \item[(2)] If the proposal $P$ has multiple projects, the mechanism randomizes over the proposed projects and no project to maximize the principal's expected payoff, while promising the agent an expected payoff of $\max_{(u,v)\in P}\alpha(u,v)u$. This is the maximal expected payoff the agent could get from proposing each project alone.
\end{enumerate}
By the definition of a PMP mechanism, the more projects the agent proposes, the weakly higher his expected payoff will be. Therefore, PMP mechanisms are incentive-compatible. Under any incentive-compatible mechanism, the agent's payoff from proposing multiple projects must be at least his maximal payoff from proposing each project alone. A PMP mechanism has the feature that the agent receives exactly his maximal payoff from proposing each project alone, but not more.

Our next theorem shows that there exists a PMP mechanism that achieves the minimal $\wcr$. Moreover, it is the agent's most preferred mechanism among all the mechanisms that achieve the minimal $\wcr$ and are admissible.

 \begin{theorem}[Multiproject environment]\label{th:multiproject}
Assume $\cale=2^D$. For every $u\in [\under u,1]$ and $p\in [0, 1]$, let $\gamma(u,p)$ be 
\begin{equation}
\label{alpha-satisfies}	
\gamma(u,p)=\min \{ q \in [0,1]: \underu+q(u-\underu) \geqslant p u\}.
\end{equation}
 Let 
     \begin{equation}
     \label{regretm}
     R^m = \max_{(u,v)\in D}\min_{p\in [0,1]}\max \left\{(1-p)v, \gamma(u,p)(1-v)\right\}.     	
     \end{equation}
\begin{enumerate}[(i)]
\vspace{-0.35cm} \item The $\wcr$ under any mechanism is at least $R^m$. 
\vspace{-0.35cm} \item Let $\rho^m$ be the PMP mechanism induced by 
    \begin{equation}\label{etam}\alpha^m(u,v)=\max\{p\in [0,1]: \gamma(u,p)(1-v)\leqslant R^m\}. \end{equation}
    Under $\rho^m$, it is optimal for the agent to propose his type $A$. The corresponding choice function $\rho^m$ has the $\wcr$ of $R^m$. It is admissible for $\under v>0$.
  \vspace{-0.35cm}  \item \label{multithree} If $\rho$ is an IC and admissible mechanism which has the $\wcr$ of $R^m$, then $U(\rho,A)\leqslant U(\rho^m,A)$ for every type $A$.
\end{enumerate}
\end{theorem}
We now explain the intuition behind $\gamma(u,p),\text{ } R^m, \text{ and }\alpha^m(u,v)$ in \eqref{alpha-satisfies} to \eqref{etam}. Pick any project $(u,v)$ and suppose $p$ is the approval probability if the agent proposes $(u,v)$ alone, so he expects a payoff of $pu$ from doing so. Then, if the agent has both $(u,v)$ and $(\underu,1)$ available, he must expect a weakly higher payoff from proposing both projects than from proposing $(u,v)$ alone. This implies that, when the agent proposes both projects, the probability of choosing the agent's preferred project $(u,v)$ is at least $\gamma (u,p)$ while the probability of choosing the principal's preferred project $(\underu,1)$ is at most $1-\gamma(u,p)$. %Naturally, the function $\gamma(u,p)$ increases in $p$. 
The principal's $\wcr$ is then at least $\max\{(1-p)v, \gamma(u,p)(1-v) \}$. The first term, $(1-p)v$, is the regret when $(u,v)$ is the only available project, while the second term, $\gamma(u,p)(1-v)$, reflects the regret when both $(u,v)$ and $(\underu,1)$ are available. The operators $\min_{p\in[0,1]}$ and $\max_{(u,v) \in D}$ in $R^m$ reflect, respectively, the fact that the principal can choose $p$ and the fact that the argument holds for any $(u,v)\in D$. Lastly, $\alpha^m(u,v)$ is the highest approval probability for $(u,v)$ when it is proposed alone, such that (i) the agent can be incentivized to propose both available projects when his type is $A=\{(u,v),(\underu,1)\}$, and (ii) the regret from type $A$ can be kept under $R^m$.

% This implies that the probability of choosing the agent's preferred project $(u,v)$ over the principal's preferred project $(\underu,1)$ is at least $\gamma(u,p)$ when the agent proposes both projects. Naturally, the function $\gamma(u,p)$ increases in $p$. 

%  and is rejected with positive probability otherwise

The explicit expressions for $R^m$ and $\alpha^m(u,v)$ are presented at the end of this subsection. It follows from expressions \eqref{alpha-satisfies} to \eqref{etam} that $\alpha^m(u,v)=1$ if $v \geqslant 1-R^m$ or $u=\under u$, and $\alpha^m(u,v)<1$ otherwise. As in the case of the single-project environment, when the agent proposes only one project, it is approved for sure if its payoff to the principal is sufficiently high. For this reason, we call a project with $v \geqslant 1-R^m$ a top-tier project and a project with $v < 1-R^m$ a bottom-tier project. Figure \ref{fg:compromise} depicts these two tiers when $\underu=\underv=0$.

When the agent proposes multiple projects, the principal promises the agent an expected payoff of $\max_{(u,v)\in P} \alpha^m(u,v)u$. In both panels of Figure \ref{fg:compromise}, each dotted curve connects all the projects that induce the same value of $\alpha^m(u,v)u$, so it can be interpreted as an ``indifference curve'' for the agent. For a project in the top tier, the principal is willing to compensate the agent his full payoff from such a project. In contrast, for a project in the bottom tier, the principal is willing to compensate the agent only a discounted payoff. The lower the project's payoff to the principal, the more severe the discounting. Hence, indifference curves are vertical in the top tier and downward sloping from left to right in the bottom tier.

The agent's expected payoff is determined by the project (among those proposed) that is on the highest indifference curve. This is the project that the agent himself wants to propose.

\begin{figure}[htb!]
\begin{center}
\psset{xunit=4.5cm,yunit=4.5cm}

\begin{pspicture}(-.31,-.33)(2.94,1.1)

%\psline(-.31,-.33)(2.94, -.33)(2.94,1.1)(-.31,1.1)(-.31,-.33)

\rput[c](.5,-.18){Left panel: Agent's favorite project}
\rput[c](.5,-.28){is in the top tier.}

\rput[c](1.4,0.892){top tier}
\psline[arrows=<-,arrowscale=1.6](1.01,.892)(1.16,.892)
\psline[arrows=->,arrowscale=1.6](1.64,.892)(1.79,.892)
\rput[c](1.4,0.375){bottom tier}
%\rput[c](1.4,0.335){tier}
\psline[arrows=<-,arrowscale=1.6](1.01,.375)(1.16,.375)
\psline[arrows=->,arrowscale=1.6](1.64,.375)(1.79,.375)

%\rput[c](.5,-.18){left panel: $\underv=0$}
%\rput[c](2,-.18){right panel: $\underu=0$}

\psaxes{->,ticks=none,labels=none}(0,0)(0,0)(1.05,1.05)

\rput[c](1.08,0){$u$}
\rput[c](1,-.06){$1$}

\rput[c](0,1.08){$v$}
\rput[r](-.02,1){$1$}
\rput[c](-.02,-.06){$0$}

\psline[linecolor=black](1,-.02)(1,.02)
\psline[linecolor=black](-.02,1)(.02,1)

\psline[linecolor=black](0,1)(1,1)(1,0)

\psline[linecolor=black,linestyle=dashed, linewidth=1.6pt](0,.75)(1,.75)
\rput[r](-.02,.75){$1-R^m$}

% indiff lines:

\psline[linecolor=black, linestyle=dotted, linewidth=1.6pt](.8,1)(0.8, 0.75)(1., 0.6875)
\psline[linecolor=black, linestyle=dotted, linewidth=1.6pt](.5,1)(0.5, 0.75)(1., 0.5)
\psline[linecolor=black, linestyle=dotted, linewidth=1.6pt](.2,1)(0.2, 0.75)(0.8, 0.)

%\psdot*[dotstyle=triangle*,dotsize=6pt](.7,.8)
%\psdot*[dotstyle=pentagon*,dotsize=6pt](.4,.9)

\rput[c](.7,.83){{\large $\blacktriangle$}}
\rput[c](.4,.93){$\bigstar$}

%\psdot*[dotstyle=triangle*,dotsize=6pt](2.2,.4)
%\psdot*[dotstyle=pentagon*,dotsize=6pt](1.55,.9)

%\rput[c](.25,0.785714){$R^s(\underu,0)$}
%\rput[c](.25,0.414624){$R^m(\underu,0)$}

%%% second half

\psaxes{->,ticks=none,labels=none}(1.8,0)(1.8,0)(2.85,1.05)

\rput[c](2.88,0){$u$}
\rput[c](1.8,-.06){$0$}
\rput[c](2.8,-.06){$1$}

\rput[c](1.8,1.08){$v$}
\rput[r](1.78,1){$1$}

\psline[linecolor=black](1.78,1)(1.82,1)
\psline[linecolor=black](2.8,-.02)(2.8,.02)

\psline[linecolor=black](1.8,1)(2.8,1)(2.8,0)

\psline[linecolor=black,linestyle=dashed,linewidth=1.6pt](1.8,.75)(2.8,.75)
\rput[r](1.78,.75){$1-R^m$}

% indiff curves:
\psline[linecolor=black,linestyle=dotted,linewidth=1.6pt](2.6,1)(2.6, 0.75)(2.8, 0.6875)
\psline[linecolor=black, linestyle=dotted, linewidth=1.6pt](2.3,1)(2.3, 0.75)(2.8,0.5)
\psline[linecolor=black,linestyle=dotted,linewidth=1.6pt](2.0,1)(2.0, 0.75)(2.6, 0.)

% the agent's preferred
\rput[c](2.05,0.6875){$\blacksquare$}
\rput[c](1.96,.95){$\bigstar$}
\rput[c](2.4,.214){{\large $\blacktriangle$}}

\rput[c](2.0,-0.05){$x$}
\psline[linestyle=dashed,linewidth=.2pt](2,.75)(2, 0)

%\rput[c](2.5,.4){{\large $\blacktriangle$}}
%\rput[c](1.9,.85){$\bigstar$}

\rput[c](2.3,-.18){Right panel: Agent's favorite project}
\rput[c](2.3,-.28){is in the bottom tier.}

\end{pspicture}
\end{center}
\vspace{-6 mm}
\caption{Optimal mechanism in the multiproject environment, $\underu=\underv=0$}
\label{fg:compromise}
\end{figure}
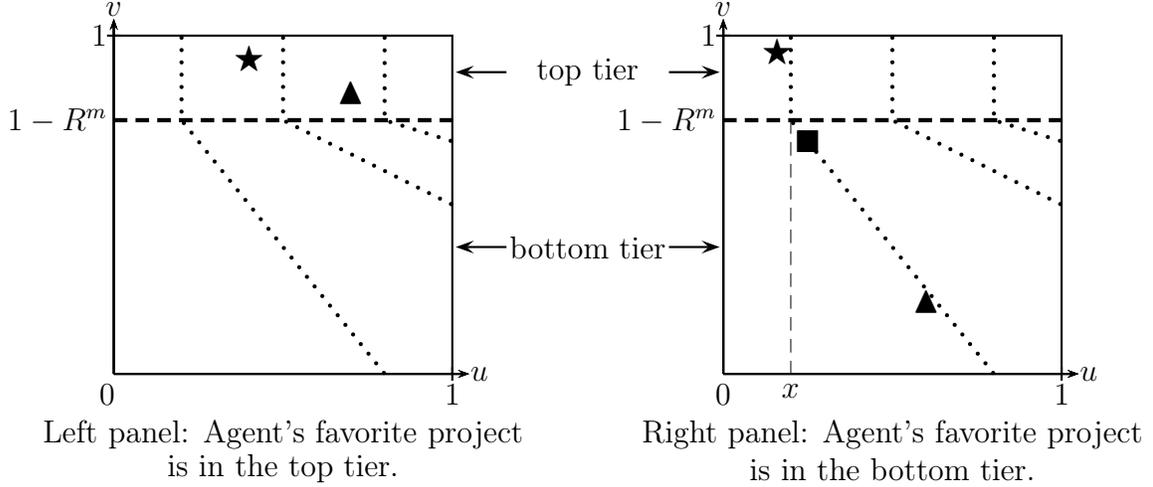

So when does the principal benefit from the agent's also proposing other projects? Figure \ref{fg:compromise} gives an intuitive answer. If the agent's favorite project is already a top-tier project for the principal, then this project will be chosen for sure, so there is no benefit from also proposing other projects. The left panel gives such an example: $\bigstar$ and $\blacktriangle$ denote the available projects and $\blacktriangle$ will be chosen for sure. In contrast, if the agent's favorite project is a bottom-tier project for the principal, the benefit from the agent's proposing other projects as well can be significant. The right panel illustrates such an example. Among the available projects, project $\blacksquare$ is on the highest indifference curve, so it pins down the maximal expected payoff the agent could get from proposing each project alone. This maximal payoff is denoted by $x$ in the figure. The mechanism randomizes between $\blacktriangle$ and $\bigstar$ while promising the agent an expected payoff of $x$. In such cases, multiple projects in the proposal have a chance of being chosen: sometimes the choice favors the agent, and at other times it favors the principal.

% Among the available projects, project $\blacksquare$ is on the highest indifference curve, so it is what the agent himself wants to propose. If the agent proposed $\blacksquare$ alone, his expected payoff would be $x$. The mechanism randomizes between $\blacktriangle$ and $\bigstar$ while promising the agent the expected payoff $x$ that he would get from proposing $\blacksquare$ alone. In such cases, multiple projects in the proposal have a chance of being chosen: sometimes the choice favors the agent, and at other times it favors the principal. 

Statement (\ref{multithree}) in Theorem \ref{th:multiproject} says that the mechanism $\rho^m$ is the agent's most preferred mechanism among all the mechanisms that have the $\wcr$ of $R^m$ and are admissible. This preference for the mechanism $\rho^m$ holds for every possible type $A$ of the agent. 

Lastly, the explicit expressions for $R^m$ and $\alpha^m$ are given below. The probability $\alpha^m(u,v)$ is increasing in $v$ and decreasing in $u$. It is strictly so in the bottom tier.
\[
R^m=\begin{cases}
 \frac{(1-\underu) \left(2-\underu-2 \sqrt{1-\underu}\right)}{\underu^2} & \text{ if }\underv<\frac{1-\sqrt{1-\underu}}{\underu}; \\
 \frac{(1-\underu) (1-\underv) \underv}{1-\underu \; \underv} & \text{ otherwise};
\end{cases}\]
and
\[
\alpha^m(u,v)= \begin{cases}1, &\text{ if }v\geqslant 1-R^m \text{ or } u=\under u;\\ \left(1-\frac{R^m}{1-v} \right) { \under u}/{u}+\frac{R^m}{1-v}, &\text{ if }v<1-R^m \text{ and } u>\under u. \end{cases}\]   
%\begin{remark}
%The PMP mechanism $\rho^m$ is weakly IC. We could modify the mechanism to make it strictly IC. To do this, let $h:[0, \infty) \rightarrow [0,1]$ be a strictly increasing function and $|P|$ the number of projects in the proposal $P$. If we modify $\alpha^m$ with $(1-\varepsilon + \varepsilon h(|P|))\alpha^m$ for some small $\varepsilon>0$, then the PMP mechanism induced by this modified function has the property that if the agent has some available project with $u>0$ then it is strictly optimal for him to propose all of them. The $\wcr$ under this modified mechanism approaches $R^m$ as $\varepsilon\rightarrow 0$. (If all available projects gives $0$ payoff to the agent then there is no way to strictly incentivize him to propose.)
%\eran{I changed this remark a bit to take care of the case that all projects have $u=0$} 
%\end{remark}

\begin{remark}
The PMP mechanism $\rho^m$ is weakly IC. We can modify this mechanism to make it strictly IC. To do this, let $h:[0, \infty) \rightarrow [0,1]$ be a strictly increasing function and $|P|$ the number of projects in the proposal $P$. We modify $\alpha^m$ with $(1-\varepsilon + \varepsilon h(|P|))\alpha^m$ for some small $\varepsilon>0$ and consider the PMP mechanism induced by this modified function. If the agent has some projects with $u>0$, then it is strictly optimal for him to propose all the available projects. The $\wcr$ under this mechanism approaches $R^m$ as $\varepsilon\rightarrow 0$. (If the agent has only projects with $u=0$, then there is no way to strictly incentivize him to propose.)
\end{remark}

\subsection{Comparing the $\wcr$ under the two environments}

According to the definitions in \eqref{r-one} and \eqref{regretm}, both $R^s$ and $R^m$ decrease in $\underu$ and $\underv$. Figure \ref{fg:cutoff} compares the values of $R^s$ and $R^m$ as $\underu$ or $\underv$ varies. The left panel depicts the $\wcr$ as a function of $\under u$ for a fixed $\under v$. The right panel depicts the $\wcr$ as a function of $\under v$ for a fixed $\under u$. Roughly speaking, the principal's gain from having the multiproject environment instead of the single-project environment, as measured by $R^s-R^m$, is larger when $\under u$ or $\under v$ is smaller (i.e., when the principal faces more uncertainty or when players can potentially have stronger preferences over projects).  

\begin{figure}[htb!]
\begin{center}
\psset{xunit=5.4cm,yunit=3.375cm}

\begin{pspicture}(-.1,-.2)(2.60,1.2)

\rput[c](.5,-.18){Left panel: $\underv=0$}
\rput[c](2,-.18){Right panel: $\underu=0$}

\psaxes{->,ticks=none,labels=none}(0,0)(0,0)(1.05,1)

\rput[c](1.08,0){$\underu$}
\rput[c](0,-.08){$0$}
\rput[c](1,-.08){$1$}

\rput[c](0,1.08){$\wcr$}
\rput[r](-.04,.8){$1/2$}
\rput[r](-.04,.4){$1/4$}

\psline[linecolor=black](-.02,.8)(0,.8)
\psline[linecolor=black](-.02,.4)(0,.4)
\psline[linecolor=black](1,-.02)(1,.02)

\rput[c](.25,0.775714){$R^s$}%{$R^s(\underu,0)$}
\rput[c](.25,0.414624){$R^m$}
%{$R^m(\underu,0)$}

\psline[linecolor=black,linestyle=dashed,linewidth=1.6pt](0., 0.8)(0.025, 0.789873)(0.05, 0.779487)(0.075,0.768831)(0.1, 0.757895)(0.125, 0.746667)(0.15,0.735135)(0.175, 0.723288)(0.2, 0.711111)(0.225,0.698592)(0.25, 0.685714)(0.275, 0.672464)(0.3,0.658824)(0.325, 0.644776)(0.35, 0.630303)(0.375,0.615385)(0.4, 0.6)(0.425, 0.584127)(0.45, 0.567742)(0.475,0.55082)(0.5, 0.533333)(0.525, 0.515254)(0.55,0.496552)(0.575, 0.477193)(0.6, 0.457143)(0.625,0.436364)(0.65, 0.414815)(0.675, 0.392453)(0.7,0.369231)(0.725, 0.345098)(0.75, 0.32)(0.775, 0.293878)(0.8,0.266667)(0.825, 0.238298)(0.85, 0.208696)(0.875,0.177778)(0.9, 0.145455)(0.925, 0.111628)(0.95,0.0761905)(0.975, 0.0390244)(1., 0.)

\psline[linecolor=black,linewidth=1.6pt](0.,0.4)(0.025,0.394953)(0.05,0.389808)(0.075,0.384562)(0.1,0.37921)(0.125,0.373749)(0.15,0.368174)(0.175,0.362479)(0.2,0.35666)(0.225,0.35071)(0.25,0.344624)(0.275,0.338396)(0.3,0.332017)(0.325,0.325481)(0.35,0.318779)(0.375,0.311902)(0.4,0.30484)(0.425,0.297583)(0.45,0.290119)(0.475,0.282435)(0.5,0.274517)(0.525,0.266348)(0.55,0.257913)(0.575,0.24919)(0.6,0.240158)(0.625,0.230792)(0.65,0.221063)(0.675,0.210938)(0.7,0.20038)(0.725,0.189344)(0.75,0.177778)(0.775,0.165618)(0.8,0.152786)(0.825,0.139189)(0.85,0.124701)(0.875,0.109164)(0.9,0.0923545)(0.925,0.0739498)(0.95,0.0534326)(0.975,0.0298234)(1.,0.)

\psaxes{->,ticks=none,labels=none}(1.5,0)(1.5,0)(2.55,1)
\rput[c](2.58,0){$\underv$}
\rput[c](1.5,-.08){$0$}
\rput[c](2.5,-.08){$1$}

\rput[c](1.5,1.08){$\wcr$}
\rput[r](1.46,.8){$1/2$}
\rput[r](1.46,.4){$1/4$}

\psline[linecolor=black](1.48,.8)(1.5,.8)
\psline[linecolor=black](1.48,.4)(1.5,.4)
\psline[linecolor=black](2.5,-.02)(2.5,.02)

\rput[c](1.75,0.86){$R^s$}%{$R^s(0,\underv)$}
\rput[c](1.75,0.46){$R^m$}
%{$R^m(0,\underv)$}

\psline[linecolor=black,linestyle=dashed,linewidth=1.6pt](1.5,0.8)(1.525,0.8)(1.55,0.8)(1.575,0.8)(1.6,0.8)(1.625,0.8)(1.65,0.8)(1.675,0.8)(1.7,0.8)(1.725,0.8)(1.75,0.8)(1.775,0.8)(1.8,0.8)(1.825,0.8)(1.85,0.8)(1.875,0.8)(1.9,0.8)(1.925,0.8)(1.95,0.8)(1.975,0.8)(2.,0.8)(2.025,0.76)(2.05,0.72)(2.075,0.68)(2.1,0.64)(2.125,0.6)(2.15,0.56)(2.175,0.52)(2.2,0.48)(2.225,0.44)(2.25,0.4)(2.275,0.36)(2.3,0.32)(2.325,0.28)(2.35,0.24)(2.375,0.2)(2.4,0.16)(2.425,0.12)(2.45,0.08)(2.475,0.04)(2.5,0.)

\psline[linecolor=black,linewidth=1.6pt](1.5,0.4)(1.525,0.4)(1.55,0.4)(1.575,0.4)(1.6,0.4)(1.625,0.4)(1.65,0.4)(1.675,0.4)(1.7,0.4)(1.725,0.4)(1.75,0.4)(1.775,0.4)(1.8,0.4)(1.825,0.4)(1.85,0.4)(1.875,0.4)(1.9,0.4)(1.925,0.4)(1.95,0.4)(1.975,0.4)(2.,0.4)(2.025,0.399)(2.05,0.396)(2.075,0.391)(2.1,0.384)(2.125,0.375)(2.15,0.364)(2.175,0.351)(2.2,0.336)(2.225,0.319)(2.25,0.3)(2.275,0.279)(2.3,0.256)(2.325,0.231)(2.35,0.204)(2.375,0.175)(2.4,0.144)(2.425,0.111)(2.45,0.076)(2.475,0.039)(2.5,0.)

\end{pspicture}
\end{center}
\vspace{-6 mm}
\caption{$\wcr$: single-project (dashed curve) vs. multiproject (solid curve)}
\label{fg:cutoff}
\end{figure}
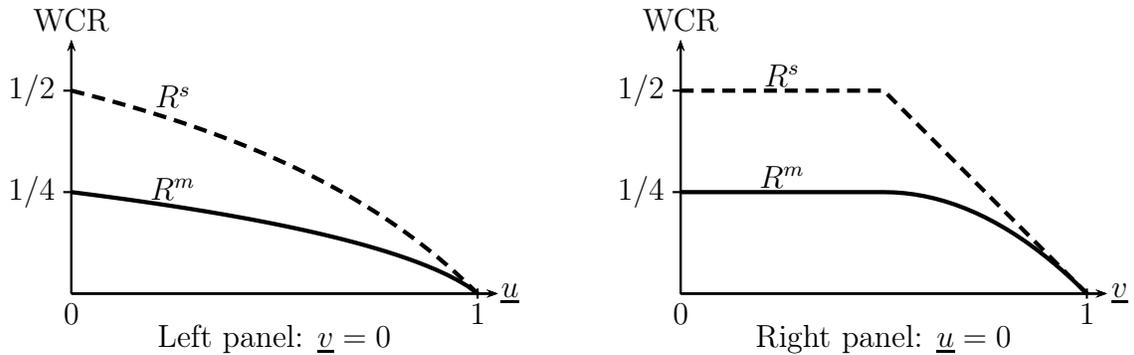

\section{Extensions and discussion}

\subsection{Intermediate environments}\label{se:intermediate}

So far, we have focused on the single-project and the multiproject environments, which are natural first steps. Nevertheless, there are intermediate environments in which the agent can propose up to $K$ projects for some fixed $K \geqslant 2$, so that $\cale=\{P\subseteq D: |P|\leqslant K \}$. We call this case the $K$-project environment. The following proposition shows that as long as $K\geqslant 2$, the principal's minimal $\wcr$ is $R^m$. 
\begin{proposition}[Two is enough in minimizing $\wcr$]\label{prop:two} Consider the $K$-project environment. For any $K\geqslant 2$,
\begin{enumerate}[(i)]
\vspace{-0.35cm}\item the $\wcr$ under any mechanism is at least $R^m$;
\vspace{-0.35cm}\item under the PMP mechanism $\rho^m$ induced by $\alpha^m(u,v)$, it is optimal for the agent to propose the principal's favorite project and the project that maximizes $\alpha^m(u,v)u$. The corresponding choice function has the $\wcr$ of $R^m$.
\end{enumerate} 
\end{proposition}

\begin{proof}
Let $A$ be the set of available projects. Let $(u_p,v_p)\in \argmax\{v:(u,v)\in A\}$ be a principal's favorite project. Let $(u_a,v_a)\in \argmax\{\alpha^m(u,v)u:(u,v)\in A\}$ be a project that gives the agent his maximal payoff from proposing each project alone. Let $P=\{(u_p,v_p), (u_a,v_a)\}$. Under the PMP mechanism $\rho^m$ induced by $\alpha^m(u,v)$, the agent is willing to propose $P$ since this proposal gives him $\alpha^m(u_a,v_a)u_a$, the maximal payoff he can get under $\rho^m$. The principal's payoff given proposal $P$ equals his payoff if the set of available projects were actually $P$. By Theorem~\ref{th:multiproject} this payoff is at least $v_p - R^m$, so the principal's regret is at most $R^m$.
\end{proof}

Based on Theorem \ref{th:one-project} and Proposition \ref{prop:two}, the principal's minimal $\wcr$ drops distinctly from $R^s$ to $R^m$ as $K$ goes from one to two or above. If the agent can propose at most one project, the only fallback option for the principal is to reject the proposed project. By contrast, if the agent can propose two or more projects, he can provide the principal with a much better fallback option than rejection. This is exactly how the principal benefits from allowing the agent to propose more than one project.

Specifically, the agent can propose the following two projects: the project that gives the agent his maximal payoff from proposing each project alone, and the principal's favorite project. The first project is what the agent himself wants to propose. The second project can be the fallback option when the first project is not chosen. According to the proof of Proposition \ref{prop:two}, the randomization between these two projects alone already ensures that the principal's $\wcr$ does not exceed $R^m$.

This result also suggests a parsimonious way to implement the PMP mechanism $\rho^m$ in the multiproject environment---a way which keeps the principal's $\wcr$ at the level of $R^m$. Instead of asking the agent to propose all the available projects, the principal can ask him to propose the aforementioned two projects: the one that gives the agent his maximal payoff from proposing each project alone and the one that the principal likes the most. This two-project implementation, however, leads to a choice function that is weakly dominated by the choice function if the agent proposed all the available projects, as shown by the right panel of Figure \ref{fg:compromise}. If the agent proposes all three projects, the mechanism randomizes between $\bigstar$ and $\blacktriangle$ while promising the agent the same payoff he would get from proposing $\blacksquare$ alone. If the agent proposes two projects, $\blacksquare$ and $\bigstar$, the mechanism randomizes between these two while promising the agent the same payoff he would get from proposing $\blacksquare$ alone; this leads to a strictly lower payoff for the principal than if the agent proposed all three projects. 

Lastly, for any $K \geqslant 3$, the principal can implement the same choice function as that in statement (ii) of Theorem \ref{th:multiproject}. This is because the linear programming problem that describes mechanism $\rho^m$ (problem~\eqref{mm1} in section~\ref{se:proofs}) has an optimal solution whose support has at most two projects. The agent can propose these two projects and the project that gives the agent his maximal payoff from proposing each project alone.

%\eran{Added the following paragraph for $K\geqslant 3$, yg: work on this.}
%However, the principal can implement the PMP mechanism $\rho^m$ if the agent is required to propose at most three projects. This is because the linear program that describes $\rho^m$ (Problem~\eqref{mm1} in Section~\ref{se:proofs}) has an optimal solution $\pi$ with support of at most two projects. Therefore, the principal can ask the agent to propose these two projects and the project that gives the agent his maximal payoff from proposing a single project. 

\subsection{Cheap-talk communication does not help for any $\cale$}  \label{se:cheaptalk}
%We could have started from a more general definition of a mechanism that chooses a project based on both the proposal $P$ and a cheap-talk message $m$ from the agent. However, in our model, cheap talk does not benefit the principal. This is because the principal can choose a project only from the proposed set $P$ and he knows the payoffs that each project in $P$ gives to both parties. Hence, there is no relevant information asymmetry left after the agent proposes $P$, and so there is no benefit to cheap talk. 

We could have started from a more general definition of a mechanism that chooses a project based on both the proposal $P$ and a cheap-talk message $m$ from the agent. However, cheap talk does not benefit the principal in our model, because the principal observes the payoffs that each project in $P$ gives to both parties. There is no relevant information asymmetry left after the agent proposes $P$, so there is no benefit to cheap talk. (If, instead, the principal does not observe the agent's payoffs from projects in $P$, then cheap talk may help.)
% (If the principal observes only the payoffs to himself but not the payoffs to the agent, then cheap talk can help.)

More specifically, for any $P$ and any cheap-talk messages $m_1,m_2$, we argue that it is without loss for the principal to choose the same subprobability measure over $P$ after $(P,m_1)$ and after $(P,m_2)$. Suppose otherwise that the principal chooses a measure $\pi_1 $ after $(P,m_1)$ and chooses $\pi_2 $ after $(P,m_2)$. If the agent strictly prefers $\pi_1$ to $\pi_2$, then he can profitably deviate to $(P,m_1)$ whenever he is supposed to say $(P,m_2)$. Hence, $(P,m_2)$ never occurs on the equilibrium path. If the agent is indifferent between $\pi_1$ and $\pi_2$, then the principal can pick his preferred measure between $\pi_1$ and $\pi_2$ after both $(P,m_1)$ and $(P,m_2)$, without affecting the agent's incentives. This argument does not depend on the exogenous restriction $\cale$ on the agent's proposal $P$, so cheap-talk communication does not help for any $\cale$.

\subsection{The commitment assumption}

Commitment is crucial for the principal to have some ``bargaining power'' in the project-choice problem. If the principal has no commitment power, sequential rationality requires that he choose his favorite project among the proposed projects. The agent now has all the bargaining power. He will propose only his favorite project, which will be chosen for sure.

In the multiproject environment, the full-commitment solution involves two types of ex post suboptimality. First, it is possible that no project is chosen even though the agent has proposed some. Second, it is possible that a worse project for the principal is chosen even though a better project for him is also proposed. Some applications may fall between the full-commitment and the no-commitment settings: the principal can commit to choosing no project but cannot commit to choosing a worse project when a better project is also proposed. In such a partial-commitment setting, a multiproject proposal is effectively a single-project proposal which contains only the principal's favorite project among the proposed ones. The optimal mechanism in this partial-commitment setting is then the same as that in the single-project environment characterized in Theorem \ref{th:one-project}. For a formal discussion of this partial-commitment environment, we refer readers to subsection \ref{sec:partial} and Theorem \ref{thm:partial}. 

\subsubsection{Commitment to a randomized mechanism}
\label{se:committorandomization} % detailed

We have shown that randomization plays a crucial role in counteracting the agent's bias to propose his favorite project and hide his less preferred ones. This raises the question of how an organization implements a randomized mechanism. In addition, if we justify commitment to randomization based on repeated interaction and reputation, a tension might exist between repeated interaction and our prior-free approach.\footnote{We thank anonymous referees for this insightful observation.} We address both points through the lens of an example from \cite{Ross1986}. \cite{Ross1986} surveyed twelve large manufacturers and described how discretionary projects are chosen through a hypothetical example. The example shows how an organization can vary the approval probabilities for different projects using its institutional rules and the reputation of the approval authority. It also illustrates the fact that building reputation through repeated interaction and being able to form a prior belief do not always go hand in hand. 

In this example, an employee has identified an approach to cutting energy costs at a heater using advanced combustion controls. He knows other approaches such as replacing the heater or adding heat exchangers. ``However the projects overlap; he can advocate at most one'' (\cite{Ross1986}, p.\ 16). Levels of approval authority vary with project sizes. The final decisions for small projects (those which cost less than $\$1$ million) are made at division headquarters. For medium projects (those between $\$1$ and $\$10$ million), they are made at the Corporate Investment Committee (CIC). The employee considers whether to propose a small project or a medium one. He has observed that going to the CIC requires ``enthusiastic support...from the division, plant, and facility managers'' (\cite{Ross1986}, p.\ 17). He has also observed that managers give higher priority to new markets or to improving the manufacturing process than to cost cutting, so enthusiastic support ``would be hard to get for a project that cuts cost but has no other production benefit'' (\cite{Ross1986}, p.\ 17). These observations imply a low approval probability for his medium-sized project. He chooses to propose the small one. 

The optimal mechanism in our single-project environment provides one possible explanation for the practice in this example. Projects in high-priority areas correspond to top-tier projects, while those in low-priority areas correspond to bottom-tier projects. Among bottom-tier projects, larger projects require broader and more enthusiastic support, so they face lower approval probabilities. This counteracts employees' tendency to propose larger projects due to empire-building incentives.

The example also suggests when reputation through repeated interaction can coexist with a prior-free approach. The principal can build a reputation for honest randomization if the agent (or an outside observer) observes previous interactions in which the principal committed to some randomized mechanism. In order for the commitment to be statistically testable, these interactions need not be with the same agent or occur in similar situations. In the example, the employee observes decisions on previous proposals by other employees, which spanned various areas and divisions. On the other hand, forming a prior belief requires some stationarity in the environment. Previous situations must be sufficiently similar to the current one, so the principal can use data on previous situations to form a prior belief about the current one. In the example, the employee is advocating a cutting-edge technology in a specific area of energy conservation, so it is unlikely that the principal has encountered many similar situations before. In fact, discretionary projects are mainly on expanded markets, new business, or new technologies for cost cutting, so they are likely to be uncharted territory for the principal. However, being in uncharted territory does not deprive the principal of his reputation for honest randomization.

\section{Conclusion}

%In this paper, we study project-choice problems in which the agent privately observes which projects are available. How to counteract the agent's bias to propose his favorite project and hide his less preferred ones is a first-order concern. We characterize prior-free mechanisms that counteract this bias, both when the agent can propose only a single project and when he can propose multiple ones. Along the way, we show that multiproject proposals play an important strategic role in providing the principal with better fallback options than rejection.  

We study project-choice problems in which the agent privately observes which projects are available. How to counteract the agent's bias to propose his favorite project and hide his less preferred ones is a first-order concern. We characterize prior-free mechanisms that counteract this bias. If the agent can propose only a single project, the optimal mechanism has a two-tier structure: a top-tier project is approved with no questions asked while a bottom-tier project will be scrutinized. The more a bottom-tier project benefits the agent, the less likely it will be approved. If the agent can propose multiple projects, he can provide the principal with better fallback options than rejection. We show that the principal does not always choose his favorite project among the proposed projects.

Our analysis shows that the minimax regret approach not only is tractable but also preserves the fundamental trade-off faced by the principal. The approach might be useful for studying other variants of the project-choice problem. For instance, the principal may not observe the agent's payoffs from proposed projects; the project-choice process may involve a deeper hierarchy involving three or more parties; the principal may be able to acquire information about the agent's type at a cost; or the agent may need to exert effort to discover projects (\cite{ArmstrongVickers2010}). We leave these questions for future research.

\section{Proofs}
\label{se:proofs}

    \subsection{Proof of Theorem~\ref{th:one-project}}

    We have proven statement (i) in the first paragraph after Theorem \ref{th:one-project}.  
   \begin{claim}
   The choice function under $\alpha^s$ in statement (ii) has the $\wcr$ of $R^s$.
   \end{claim}
  \begin{proof}We call a project $(u,v)$ \emph{top-tier} if $v\geqslant 1-R^s$ and \emph{bottom-tier} if $v<1-R^s$. 

Under $\alpha^s$, the agent's expected payoff is $\underu$ if he proposes a bottom-tier project, and is $u\geqslant \underu $ if he proposes a top-tier project $(u,v)$. Therefore, if the agent has some top-tier project, it is optimal to propose his favorite top-tier project $(u,v)$, in which case the regret is at most $1-v\leqslant R^s$. If all the projects in $A$ are bottom-tier projects, it is optimal to propose a project $(u',v')$ that maximizes $\alpha^s(u,v)v$. Let $(u_p,v_p)\in \argmax_{(u,v)\in A} v$ be a principal's favorite project in $A$. Since this is a bottom-tier project, it follows from the definition of $R^s$ in \eqref{r-one} that $(1-\under u)v_p \leqslant R^s$. Thus, the regret is $v_p- \alpha^s(u',v')v' \leqslant v_p-\alpha^s(u_p,v_p)v_p \leqslant (1-\underu)v_p\leqslant R^s$.
 \end{proof}
\begin{claim}If $\alpha$ implements a choice function that has the $\wcr$ of $R^s$, then $\alpha(u,v)\leqslant \alpha^s(u,v)$ for every $(u,v)\in D$.  \label{claim:lowerprob}
\end{claim}
%\eran{I divided a claim to two claims}
\begin{proof}Fix a project $(u,v)$. If $v\geqslant1-R^s$ or $u=0$, then $\alpha^s(u,v)=1$ and therefore $\alpha(u,v)\leqslant \alpha^s(u,v)$. If $v<1-R^s$ and $u>0$, then since the $\wcr$ under $\alpha$ is $R^s$, it must be the case that if $A=\{(u,v), (\under u,1)\}$, then the agent proposes the project $(\underu, 1)$. Otherwise, the regret is at least $1-v>R^s$. Therefore $\alpha(u,v)u\leqslant \alpha(\under u,1)\under u\leqslant \under u$, which implies $\alpha(u,v)\leqslant \underu/u=\alpha^s(u,v)$, as desired.\end{proof}
\begin{claim} 
The choice function under $\alpha^s$ in statement (ii) is admissible.
\end{claim}

\begin{proof} %Suppose that $\alpha$ implements a choice function that has the $\wcr$ of $R^s$. By Claim \ref{claim:lowerprob}, $\alpha(u,v) \leqslant \alpha^s(u,v)$ for every $(u,v) \in D$. Let $Z=\{(u,v)\in D:\alpha(u,v)<\alpha^s(u,v)\}$ be the set of projects such that $\alpha$ assigns a lower approval probability than $\alpha^s$ does. If there exists some $(u,v)\in Z$ such that $v>0$, then the regret is strictly higher under $\alpha$ than under $\alpha^s$ if $A=\{(u,v)\}$. Suppose that all the projects in $Z$ have $v=0$, so $\alpha(u,v)=\alpha^s(u,v)$ for every $(u,v)$ with $v>0$.  For every $A$ which has some project with $v>0$, the choice function in statement (ii) gives the principal a weakly higher payoff than any other choice function under $\alpha$. This is because, according to the choice function in statement (ii), among all projects in $A$ that maximize the agent's expected payoff, the agent proposes one that maximizes the principal's expected payoff. For every $A$ which has only projects with $v=0$, the regret is zero under both $\alpha$ and $\alpha^s$. Hence, the choice function in statement (ii) is admissible. 

Suppose that $\alpha$ implements a choice function that has the $\wcr$ of $R^s$. By Claim \ref{claim:lowerprob}, $\alpha(u,v) \leqslant \alpha^s(u,v)$ for every $(u,v) \in D$. If $\alpha(u,v)<\alpha^s(u,v)$ for some top-tier project $(u,v)$, then if the agent's type is $\{(u,v)\}$, $\alpha$ does strictly worse than $\alpha^s$ for the principal. Now suppose that $\alpha(u,v)=\alpha^s(u,v)$ for every top-tier project $(u,v)$. Then for every type $A$ with at least one top-tier project, $\alpha$ and $\alpha^s$ implement the same project choice, which is the agent's favorite top-tier project in $A$. Hence, we only need to consider those $A$ with only bottom-tier projects. Since $\alpha(u,v) \leqslant \alpha^s(u,v)$ and the agent proposes a project in $A$ to maximize the principal's expected payoff $\alpha^s(u,v)v$ under the choice function in (ii), $\alpha$ cannot do better than the choice function in (ii) for every type $A$ with only bottom-tier projects.

%
%
%Lastly, we show that the choice function in statement (ii) is admissible. Fix a mechanism $\rho$ which has the $\wcr$ of $R^s$. Consider the set of all proposals with exactly one top-tier project. If there exists one such proposal $P$ such that $\rho$ does not choose the top-tier project with probability $1$, then if the agent's type is actually $P$, $\rho$ does strictly worse than $\rho^*$ for the principal. 
%
%Now suppose that, for every proposal $P$ with exactly one top-tier project, $\rho$ chooses the top-tier project with probability one. Then $\rho$ and $\rho^\ast$ implement the same project choice for every type $A$ with at least one top-tier project, which is the agent's favorite top-tier project in $A$. Hence, we only need to consider the set of types with only bottom-tier projects. Since $\rho((\cdot,v_p(P))|P) \leqslant \rho^\ast((\cdot,v_p(P))|P)$ and the agent chooses a project in $A$ to maximize the principal's expected payoff under $\rho^*$, $\rho$ cannot do better than $\rho^\ast$ for every type $A$ with only bottom-tier projects. 

\end{proof}

\subsection{Proof of Theorem~\ref{th:multiproject}}
Let $a^\ast=(\underu, 1)$. Let $\overU(P)$ be the optimal value of the following linear programming with variables $\pi(u,v)$ for every $(u,v)\in P$:
\begin{maxi!}[2]{\pi}{\under u+\sum_{(u,v)\in P}\pi(u,v) (u-\under u)}{\label{overu}}{\overU(P)=}\addConstraint{\pi(u,v)}{\geqslant0,\; \forall (u,v) \in P}\addConstraint{\sum_{(u,v)\in P}\pi(u,v)}{\leqslant 1\label{overuconst1}}\addConstraint{\sum_{(u,v)\in P}\pi(u,v)(1-v)}{\leqslant R^m. \label{overuconst2}}
\end{maxi!}
The following claim explains the role of $\overU(P)$ in our argument: $\overU(P)$ is the maximal payoff that the principal can give the agent for proposing $P$ such that the principal can give the agent this same payoff if the agent proposed $P\cup \{a^\ast\}$, while still keeping regret below $R^m$.

\begin{claim}\label{cl:whyover}If $\rho$ is an IC mechanism which has the $\wcr$ of at most $R^m$, then $U(\rho,P)\leqslant \overU(P)$ for every proposal $P$.\end{claim}
\begin{proof}Let $\tilde P=P\cup\{a^\ast\}$. Let $\pi=\rho(\cdot|\tilde P)$. Since the regret under the mechanism $\rho$ when the set of available projects is $\tilde P$ is at most $R^m$, it follows that $\sum_{(u,v)\in P}\pi(u,v)(1-v)\leqslant R^m$. Therefore the restriction of $\pi$ to the set $P$ is a feasible point in problem~\eqref{overu}. Moreover
\begin{equation}\label{moreover}
U(\rho,\tilde P)=\pi(a^\ast)\underu+\sum_{(u,v)\in P}\pi(u,v)u\leqslant \underu+\sum_{(u,v)\in P}\pi(u,v) (u-\under u),\end{equation}
where the inequality follows from $\pi(a^\ast)+\sum_{(u,v)\in P}\pi(u,v)\leqslant 1$. The right hand side of~\eqref{moreover} is the objective function of~\eqref{overu} at $\pi$. Therefore, $U(\rho,\tilde P)\leqslant \overU(P)$. Finally, since the mechanism $\rho$ is IC, it follows that $U(\rho,P)\leqslant U(\rho,\tilde P)$. Therefore, $U(\rho,P)\leqslant \overU(P)$, as desired.
\end{proof}
When $P$ is a singleton $\{(u,v)\}$, we also denote $\overU(\{(u,v)\})$ by $\overU(u,v)$. The following claim explains the role of the function $\alpha^m(u,v)$ in our argument: This is the highest possible approval probability for $(u,v)$ when it is proposed alone, such that the principal can incentivize the agent to propose both available projects when his type is $A=\{(u,v),(\underu,1)\}$, while keeping the regret from type $A$ below $R^m$.
\begin{claim}\label{cl:singleton}When $P=\{(u,v)\}$ is a singleton, $\overU(u,v)=\alpha^m(u,v)u$.\end{claim}
%\eran{I changed this proof}
\begin{proof} 
We first prove that $\overU(u,v)\leqslant \alpha^m(u,v)u$. 
Let $\pi=\pi(u,v)$ be a feasible point in problem~\eqref{overu} with $P=\{(u,v)\}$, so that $\pi(1-v)\leqslant R^m$. 
Let $p\in [0,1]$ be such that $pu=\underu +\pi (u-\under u)$. 
Then $\gamma(u,p)\leqslant\pi$ by the definition of $\gamma$ in \eqref{alpha-satisfies}. Therefore  
$\gamma(u,p)(1-v)\leqslant \pi(1-v)\leqslant R^m$. 
Therefore $\alpha^m(u,v)\geqslant p$ by the definition of $\alpha^m$ in \eqref{etam}. Therefore, $\underu +\pi (u-\under u)=pu\leqslant \alpha^m(u,v)u$. Since this holds for every feasible $\pi$, it follows that $\overU(u,v)\leqslant \alpha^m(u,v)u$, as desired.

 We now prove that $\overU(u,v)\geqslant \alpha^m(u,v)u$. Let $p$ be such that $\gamma(u,p)(1-v)\leqslant R^m$ and let $\pi=\gamma(u,p)$. Then $\pi$ is feasible in problem~\eqref{overu} and therefore $\overU(u,v)\geqslant \under u+\pi(u-\under u)\geqslant pu$ by the definition of $\gamma$ in \eqref{alpha-satisfies}.  Since this holds for every  $p$ such that $\gamma(u,p)(1-v)\leqslant R^m$, it follows by the definition of $\alpha^m$ in \eqref{etam} that $\overU(u,v)\geqslant \alpha^m(u,v)u$, as desired.
\end{proof}

For a proposal $P$, let $\underU(P)=\max_{(u,v)\in P}\alpha^m(u,v)u$.
The following claim explains the role of $\underU(P)$ in our argument.
\begin{claim}\label{cl:whyunder}If $\rho$ is an IC mechanism that approves the singleton proposal $\{(u,v)\}$ with probability $\alpha^m(u,v)$, then $U(\rho,P)\geqslant\underU(P)$.\end{claim}
\begin{proof}Since $\rho$ is IC, $U(\rho,P)\geqslant U(\rho,\{(u,v)\})=\alpha^m(u,v)u$ for every $(u,v)\in P$. \end{proof}
Claim~\ref{cl:whyover} bounds from above the agent's expected payoff in an IC mechanism which has the $\wcr$ of at most $R^m$. Claim~\ref{cl:whyunder} bounds from below the agent's expected payoff in an IC mechanism which approves the singleton proposal $\{(u,v)\}$ with probability $\alpha^m(u,v)$. The following claim shows that the definition of $R^m$ is such that both bounds can be satisfied.
\begin{claim}\label{underover}$\underU(P)\leqslant \overU(P)$ for every $P$.
\end{claim}
\begin{proof}The function $\overU(P)$ defined in~\eqref{overu} is increasing in $P$. Therefore, from Claim~\ref{cl:singleton} we have:  
\[\alpha^m(u,v)u=\overU(u,v)\leqslant \overU(P), \; \forall (u,v) \in P.\]
It follows that: 
\[\underU(P)=\max_{(u,v)\in P}\alpha^m(u,v)u \leqslant\overU(P).\]
\end{proof}
By definition, $\rho^m(\cdot|P)$ solves the following linear programming:
\begin{maxi!}[2]{\pi}{\sum_{(u,v)\in P}\pi(u,v) v} {\label{mm1}} {}\addConstraint{\pi(u,v)}{\geqslant0,\; \forall (u,v) \in P}\addConstraint{\sum_{(u,v)\in P}\pi(u,v)}{\leqslant 1\label{mm1const1}}\addConstraint{\sum_{(u,v)\in P}\pi(u,v)u}{=\underU(P).\label{mm1const2}}\end{maxi!}

It is possible that \eqref{mm1} has multiple optimal solutions. Since all the optimal solutions are payoff-equivalent for both the principal and the agent, we do not distinguish among them. From now on, the notation $\rho(\cdot|P) \neq \rho^m(\cdot|P)$ means that $\rho(\cdot|P)$ is not among the optimal solutions to \eqref{mm1}. We now show that, when there is only one available project, the regret under $\rho^m$ is at most $R^m$.   
\begin{claim}\label{cl:optimal-singleton}For every singleton $A=\{(u,v)\}$, the regret under $\rho^m$ is at most $R^m$.\end{claim}
\begin{proof}Since $\rho^m$ approves $\{(u,v)\}$ with probability $\alpha^m(u,v)$, the regret is $v(1-\alpha^m(u,v))$. By the definition of $R^m$, there exists some $\bar p \in [0,1]$ such that $\max\left\{v(1-\bar p), \gamma(u,\bar p)(1-v)\right\} \leqslant R^m$. By~\eqref{etam}, $\bar p \leqslant \alpha^m(u,v)$. Therefore, it also follows that $v(1-\alpha^m(u,v))\leqslant v(1-\bar p)\leqslant R^m$.\end{proof}

\begin{claim}\label{cl:optimal}The optimal value in problem~\eqref{mm1} is at least $\max_{(u,v)\in P}v-R^m$.
\end{claim}
%\eran{I changed the proof of this claim. This is not because of the reviewers, I want it not to rely on Lemma~\eqref{le:thelemma} so that it will be clear that the lemma is needed only for the third part of the theorem. If you have time to check one thing, check this}
\begin{proof}It is sufficient to prove that the optimal values in each of the two problems derived from~\eqref{mm1} by replacing the equality constraint~\eqref{mm1const2} with inequalities $\leqslant$ and $\geqslant$ are  at least  $\max_{(u,v)\in P}v-R^m$. 

We first consider problem \eqref{mm1} with~\eqref{mm1const2} being replaced by: 
\begin{equation}\label{mm1leq}\sum_{(u,v)\in P}\pi(u,v)u\leqslant\underU(P).\end{equation}
Let $(u_p,v_p)\in P$ be a principal's favorite project. Let $\pi$ be the subprobability over $P$ that is given by $\pi(u,v)=\alpha^m(u_p,v_p)\delta_{(u_p,v_p),(u,v)}$
 where $\delta$ is Kronecker notation so \[\delta_{(u_p,v_p),(u,v)}=\begin{cases}1, &\text{if }(u,v)=(u_p,v_p),\\0, &\text{otherwise.}\end{cases}\]
Then $\sum_{(u,v)\in P}\pi(u,v)u=\alpha^m(u_p,v_p)u_p\leqslant \underU(P)$ by the definition of $\underU(P)$,  so $\pi$ satisfies~\eqref{mm1leq}. Also $v_p(1-\alpha^m(u_p,v_p))\leqslant R^m$ by Claim~\ref{cl:optimal-singleton}. Therefore, the value of the objective function in~\eqref{mm1} at $\pi$ is at least $v_p -R^m$, as desired.

We then consider problem \eqref{mm1} with~\eqref{mm1const2} being replaced by: 
\begin{equation}\label{mm1geq}\sum_{(u,v)\in P}\pi(u,v)u\geqslant\underU(P).\end{equation}
Let $(u_a,v_a)\in P$ be such that $(u_a,v_a)\in\argmax_{(u,v)\in P} \alpha^m(u,v) u$. Let $\pi$ be the probability distribution over $P$ such that 
 \[\pi(u,v)=\gamma(u_a,\alpha^m(u_a,v_a))\delta_{(u_a,v_a),(u,v)}+\bigl(1-\gamma(u_a,\alpha^m(u_a,v_a))\bigr)\delta_{(u_p,v_p),(u,v)}.\]
Then  
\[\sum_{(u,v)\in P}\pi(u,v)u=\gamma(u_a,\alpha^m(u_a,v_a))u_a+(1-\gamma(u_a,\alpha^m(u_a,v_a))u_p \geqslant \alpha^m(u_a,v_a)u_a=\underU(P),\] where the inequality follows from $u_p\geqslant \underu$ and \eqref{alpha-satisfies}, and the last equality follows from the definitions of $\underU(P)$ and $(u_a,v_a)$. Hence, $\pi$ satisfies~\eqref{mm1geq}. Also
\[v_p-\sum_{(u,v)\in P}\pi(u,v)v=\gamma(u_a,\alpha^m(u_a,v_a))(v_p-v_a)\leqslant \gamma(u_a,\alpha^m(u_a,v_a))(1-v_a)\leqslant R^m,\]
where the last inequality follows from the definition of $\alpha^m$. Therefore, the value of the objective function in~\eqref{mm1} at $\pi$ is at least $v_p-R^m$, as desired.
\end{proof} 

The following lemma gives an equivalent characterization of the mechanism $\rho^m$.
\begin{lemma}\label{le:thelemma}
The optimal solutions to \eqref{mm1} and those to the following problem coincide.  
%\begin{argmaxi!}[2]{\pi}{\sum_{(u,v)\in P}\pi(u,v) v}{\label{mm2}}{\rho(\cdot|P)\in} \addConstraint{\pi(u,v)}{\geqslant0,\; \forall (u,v) \in P}\addConstraint{\sum_{(u,v)\in P}\pi(u,v)}{\leqslant 1\label{mm2const1}}\addConstraint{\sum_{(u,v)\in P}\pi(u,v)u}{\geqslant\underU(P)\label{mm2const2}}\addConstraint{\sum_{(u,v)\in P}\pi(u,v)u}{\leqslant\overU(P).\label{mm2const3}}\end{argmaxi!}
\begin{maxi!}[2]{\pi}{\sum_{(u,v)\in P}\pi(u,v) v}{\label{mm2}}{} \addConstraint{\pi(u,v)}{\geqslant0,\; \forall (u,v) \in P}\addConstraint{\sum_{(u,v)\in P}\pi(u,v)}{\leqslant 1\label{mm2const1}}\addConstraint{\sum_{(u,v)\in P}\pi(u,v)u}{\geqslant\underU(P)\label{mm2const2}}\addConstraint{\sum_{(u,v)\in P}\pi(u,v)u}{\leqslant\overU(P).\label{mm2const3}}\end{maxi!}
\end{lemma}
\begin{proofof}{Lemma ~\ref{le:thelemma}}
We discuss two cases separately. 
\begin{case}Assume that there exists some $(u,v)\in P$ such that $v\geqslant1-R^m$. Consider the following linear programming which is a relaxation of both problem~\eqref{mm1} and problem~\eqref{mm2}: 
%\begin{argmaxi!}[2]{\pi}{\sum_{(u,v)\in P}\pi(u,v) v} {\label{relax}} {\rho(\cdot|P)\in}\addConstraint{\pi(u,v)}{\geqslant0,\; \forall (u,v) \in P}\addConstraint{\sum_{(u,v)\in P}\pi(u,v)}{\leqslant 1\label{relaxconst1}}\addConstraint{\sum_{(u,v)\in P}\pi(u,v)u}{\geqslant\underU(P).\label{relaxconst2}}\end{argmaxi!}
\begin{maxi!}[2]{\pi}{\sum_{(u,v)\in P}\pi(u,v) v} {\label{relax}} {}\addConstraint{\pi(u,v)}{\geqslant0,\; \forall (u,v) \in P}\addConstraint{\sum_{(u,v)\in P}\pi(u,v)}{\leqslant 1\label{relaxconst1}}\addConstraint{\sum_{(u,v)\in P}\pi(u,v)u}{\geqslant\underU(P).\label{relaxconst2}}\end{maxi!}
We claim that the constraint~\eqref{relaxconst2} holds with equality at every optimal solution. Indeed, if~\eqref{relaxconst2} is not binding then an optimal solution to \eqref{relax} is also an optimal solution to the following linear programming:
%\begin{argmaxi!}[2]{\pi}{\sum_{(u,v)\in P}\pi(u,v) v} {\label{relax2}} {\rho(\cdot|P)\in}\addConstraint{\pi(u,v)}{\geqslant0,\; \forall (u,v) \in P}\addConstraint{\sum_{(u,v)\in P}\pi(u,v)}{\leqslant 1,\label{relax2const1}}\end{argmaxi!}
\begin{maxi!}[2]{\pi}{\sum_{(u,v)\in P}\pi(u,v) v} {\label{relax2}} {}\addConstraint{\pi(u,v)}{\geqslant0,\; \forall (u,v) \in P}\addConstraint{\sum_{(u,v)\in P}\pi(u,v)}{\leqslant 1,\label{relax2const1}}\end{maxi!}
which is derived from~\eqref{relax} by removing~\eqref{relaxconst2}. Let $v_p = \max_{(u,v)\in P} v$ and $u_p = \max_{(u,v_p)\in P} u$. Since $v_p\geqslant1-R^m$, it follows from  the definition of $\alpha^m$ in \eqref{etam} that $\alpha^m(u_p,v_p)=1$. Every optimal solution $\pi$ to problem~\eqref{relax2} satisfies $\support(\pi)\subseteq \argmax_{(u,v)\in P}v$, which implies that: 
\[
\sum_{(u,v)\in P}\pi(u,v)u\leqslant u_p=\alpha^m(u_p,v_p)u_p\leqslant \underU(P).
\] 
We proved that every optimal solution to \eqref{relax} satisfies~\eqref{relaxconst2} with equality, so it is a feasible point in both \eqref{mm1} and \eqref{mm2}. Since problem~\eqref{relax} is a relaxation of both problem~\eqref{mm1} and~\eqref{mm2}, the optimal values of \eqref{mm1}, \eqref{mm2}, and \eqref{relax} coincide. %Hence, every optimal solution to \eqref{mm1} or \eqref{mm2} is optimal in \eqref{relax}. This, combined with the fact that every optimal solution to \eqref{relax} is optimal in \eqref{mm1} and \eqref{mm2}, implies that the optimal solutions to \eqref{mm1} and \eqref{mm2} coincide. 
\end{case}
\begin{case}
Assume now that $v<1-R^m$ for every $(u,v)\in P$. We claim that $\underU(P)=\overU(P)$ and therefore problems~\eqref{mm1} and~\eqref{mm2} coincide. Given that $v<1-R^m$ for every $(u,v) \in P$, the constraint~\eqref{overuconst1} in problem~\eqref{overu} must be slack since if it is satisfied with an equality then~\eqref{overuconst2} is violated. Therefore, in this case $\overU(P)$ also satisfies
\begin{maxi}[2]{\pi}{\under u+\sum_{(u,v)\in P}\pi(u,v) (u-\under u)}{\label{overu2}}{\overU(P)=}\addConstraint{\pi(u,v)}{\geqslant0,\; \forall (u,v) \in P}\addConstraint{\sum_{(u,v)\in P}\pi(u,v)(1-v)}{\leqslant R^m,}\end{maxi}
which is derived from problem~\eqref{overu} by removing~\eqref{overuconst1}. Problem~\eqref{overu2} admits a solution $\pi^\ast$ with the property that, for some $(u^\ast,v^\ast)\in P$, the only non-zero element of $\pi^\ast$ is $\pi^\ast(u^\ast,v^\ast)$. Therefore, by Claim~\ref{cl:singleton},
\[\overU(P)=\overU({u^\ast,v^\ast})=\alpha^m(u^\ast,v^\ast) u^\ast\leqslant \underU(P).\]
Therefore, by Claim~\ref{underover} we get $\overU(P)=\underU(P)$, as desired.
\end{case}
\end{proofof}

\begin{proofof}{Theorem~\ref{th:multiproject}}
\begin{enumerate}
\vspace{-0.35cm} 
\item Fix $(u,v)\in D$ and let $P=\{(u,v)\}$ and $\tilde P=\{(u,v),(\under u,1)\}$.  Let $p$ be the probability that $\rho$ chooses $(u,v)$ when the proposal is $P$. So, $\rgrt(\rho,P)=(1-p)v$. Since the mechanism is IC, the agent's expected payoff under $\tilde P$ must be at least $pu$. By definition of $\gamma(u,p)$, this implies that when the proposal is $\tilde P$ the mechanism chooses $(u,v)$ with probability at least $\gamma(u,p)$. So, $\rgrt(\rho,\tilde P)\geqslant \gamma(u,p)(1-v)$. Therefore, the principal's $\wcr$ is at least $\max\left\{(1-p)v,  \gamma(u,p)(1-v)\right\}$. Since the principal can choose $p$, %his $\wcr$ is at least the value of $\max\left\{(1-p)v,  \gamma(u,p)(1-v)\right\}$ \textit{even when} he can choose the most favorable $p$, so:
it follows that
$$
\wcr(\rho) \geqslant \min_{p \in [0,1]} \max \left\{(1-p)v,  \gamma(u,p)(1-v)\right\}.
$$
This inequality must hold for any $(u,v) \in D$. Hence, the $\wcr$ under any mechanism is at least $R^m$. 
\vspace{-0.35cm} \item The mechanism $\rho^m$ is IC, and it solves problem~\eqref{mm1}. By Claim~\ref{cl:optimal}, the optimal value in problem~\eqref{mm1} is at least $\max_{(u,v)\in P}v-R^m$. Since the objective function in \eqref{mm1} is the principal's payoff under $\rho^m$, the principal's regret is at most $R^m$.

We next argue that $\rho^m$ is admissible. Let $\rho$ be an IC mechanism which has the $\wcr$ of $R^m$ and let $\alpha(u,v)$ be the probability that $\rho$ approves a singleton proposal $\{(u,v)\}$. By Claims~\ref{cl:whyover} and~\ref{cl:singleton}, $\alpha(u,v)\leqslant\alpha^m(u,v)$. Assume first that $\alpha(u,v)<\alpha^m(u,v)$ for some project $(u,v)$. Then, since $v>0$, if the agent's type  is the  singleton $\{(u,v)\}$ the principal's payoff is strictly higher under $\rho^m$ than under $\rho$. Assume now that $\alpha(u,v)=\alpha^m(u,v)$ for every $(u,v)$. We claim that for every proposal $P$ the principal's payoff is weakly higher under $\rho^m$ than under $\rho$. Indeed, since $\rho$ is IC, it follows from Claim~\ref{cl:whyover} that $U(\rho,P)\leqslant \overU(P)$, and, from Claim~\ref{cl:whyunder}, that $U(\rho,P)\geqslant\underU(P)$. Therefore $\rho(\cdot|P)$ is a feasible point in problem~\eqref{mm2}. Since by Lemma \ref{le:thelemma} $\rho^m(\cdot|P)$ is the optimal solution to \eqref{mm2}, the principal's payoff is weakly higher under $\rho^m$ than under $\rho$.
\vspace{-0.35cm} \item Let $\rho$ be an IC, admissible mechanism which has the $\wcr$ of $R^m$ and which differs from $\rho$. We want to show that $U(\rho,P) \leqslant U(\rho^m,P)$ for every finite $P \subseteq D$. Recall that $U(\rho^m,P)=\underU(P)$ for every $P$. 

We first construct a new mechanism $\tilde\rho$ based on $\rho$ and $\rho^m$:
  \[
  \tilde\rho(\cdot|P)=\begin{cases}\rho^m(\cdot|P),&\text{ if }U(\rho,P)\geqslant \underU(P) \\\rho(\cdot|P),&\text{ if }U(\rho,P)<\underU(P).\end{cases}\]
  By definition, $U(\tilde\rho,P)=\min\left\{U(\rho,P),U(\rho^m,P)\right\}$. 
  The functions $U(\rho,P)$ and $U(\rho^m,P)$ are increasing in $P$ since $\rho$ and $\rho^m$ are IC. Therefore $U(\tilde\rho,P)$ is increasing in $P$, so $\tilde\rho$ is also IC. Moreover, for every $P$ either $\tilde \rho(\cdot|P)=\rho(\cdot|P)$ or $\tilde \rho(\cdot|P)=\rho^m(\cdot|P)$. Therefore the $\wcr$ under $\tilde\rho$ is also $R^m$.

We next argue that for every $P$, $\tilde \rho$ gives the principal a weakly higher payoff than $\rho$ does. 
\begin{enumerate}
\vspace{-0.35cm} \item Consider a set $P$ such that $U(\rho,P)<\underU(P)$. Then $  \tilde\rho(\cdot|P)=  \rho(\cdot|P)$, so $\tilde \rho$ gives the principal the same payoff as $\rho$ does. 
\vspace{-0.35cm} \item Consider a set $P$ such that $U(\rho,P) \geqslant \underU(P)$. From Claim~\ref{cl:whyover} we know that $U(\rho,P)\leqslant \overU(P)$ for every $P$. Therefore, $\rho(\cdot|P)$ is a feasible point in problem \eqref{mm2}. It follows from Lemma \ref{le:thelemma} that $\rho^m$ gives the principal a weakly higher payoff than $\rho$ does. Moreover, if $\rho(\cdot|P) \neq \rho^m (\cdot|P)$, then $\rho^m$ gives the principal a strictly higher payoff than $\rho$ does. 

Since $\tilde \rho(\cdot|P)=\rho^m(\cdot|P)$ for every $P$ such that $U(\rho,P) \geqslant \underU(P)$, $\tilde \rho$ gives the principal a weakly higher payoff than $\rho$ does for every such $P$. 
\end{enumerate}
We have argued that $\tilde \rho$ gives the principal a weakly higher payoff than $\rho$ does for every $P$. On the other hand, $\rho$ is admissible, so there cannot be a $P$ such that $\tilde \rho$ gives the principal a strictly higher payoff than $\rho$ does. This implies that $\rho(\cdot|P)=\rho^m (\cdot|P)$ for every $P$ such that $U(\rho,P) \geqslant \underU(P)$, so $U(\rho,P)=U(\rho^m,P)=\underU(P)$. Hence $U(\rho,P) \leqslant \underU(P)$ for every $P$. 
\end{enumerate}
\end{proofof}

\subsection{Partial commitment in the multiproject environment}\label{sec:partial}
%\eran{I would delete the assertion about admissibility and part iii from the theorem. In my view the interesting part is the  equivalence between partial commitment and single project under WCR. The rest just blurs this message}
Let $v_p(P)=\max_{(u,v)\in P} v$ be the principal's payoff from his favorite project in $P$. We consider an environment in which the principal cannot commit to choosing a worse project when a better project is also proposed. This means that the mechanism $\rho$ must satisfy $\rho((u,v)|P)=0$ for any $(u,v)\in P$ such that $v \neq v_p(P)$. We now show that this environment is essentially equivalent to the single-project environment.

Let $u_p(P)=\min_{(u,v) \in P, v=v_p(P)} u$ be the agent's lowest payoff among the principal's favorite projects in $P$. We now define a mechanism $\rho^\ast$ which is adapted from the optimal mechanism $\alpha^s$ in the single-project environment. The mechanism $\rho^\ast$ chooses $(u_p(P),v_p(P))$ with probability $\alpha^s(u_p(P),v_p(P))$ and chooses other projects with zero probability:  
$$
\rho^\ast((u,v)|P) = \begin{cases}
 \alpha^s(u_p(P),v_p(P)) & \text{ if } (u,v) = (u_p(P),v_p(P)) ; \\
   0 & \text{ if } (u,v) \neq (u_p(P),v_p(P)) . 	
 \end{cases}
$$
For a mechanism $\rho$, we let $\rho((\cdot,v_p(P))|P)$ denote the probability that a principal's favorite project in $P$ is chosen when the agent proposes $P$.  
\begin{theorem}\label{thm:partial} Consider the multiproject environment with partial commitment.
\begin{enumerate}[(i)]
\vspace{-0.35cm} \item The $\wcr$ under any mechanism $\rho$ is at least $R^s$. 
\vspace{-0.35cm} \item Under $\rho^\ast$, if the agent has projects with $v\geqslant 1-R^s$, it is optimal for him to propose his favorite project among those with $v\geqslant 1-R^s$. Otherwise, it is optimal to propose a project that maximizes the principal's expected payoff $\alpha^s(u,v)v$. The corresponding choice function has the $\wcr$ of $R^s$ and is admissible. 
\vspace{-0.35cm} \item \label{singlethree} If a mechanism $\rho$ implements a choice function that has the $\wcr$ of $R^s$, then 
$$
\rho((\cdot,v_p(P))|P) \leqslant \rho^\ast((\cdot,v_p(P))|P), \; \text{ for every } P . $$ 
\end{enumerate}	
\end{theorem}
\begin{proof}
\begin{enumerate}[(i)]
\vspace{-0.35cm} \item Fix $v\in [\under v,1)$. Let $x=\rho((1,v)|\{(1,v)\})$ be the probability that $\rho$ chooses $(1,v)$ if the agent proposes only one project $(1,v)$. 

If $x \leqslant \under u$, then, if the agent has only the project $(1,v)$, the regret is $(1-x)v\geqslant (1-\under u)v$. Now suppose that $x> \under u$. If the agent has two projects $(1,v)$ and $(\under u,1)$, his payoff is at most $\under u$ if he proposes both projects. This is because $\rho$ must assign zero probability to $(1,v)$ when $(\under u,1)$ is also proposed. Moreover, his payoff is at most $\under u$ if he proposes only $(\under u,1)$. On the other hand, his payoff is strictly higher than $\under u$ if he proposes only $(1,v)$ since $x>\under u$. Hence, the agent proposes only $(1,v)$ and the regret is $1-\alpha(1,v)v\geqslant 1-v$. Therefore, $\wcr\geqslant \min\left\{(1-\under u)v, 1-v\right\}$ for every $v\in [\under v,1)$.	
% search me 
\vspace{-0.35cm} \item  We call a project $(u,v)$ \emph{top-tier} if $v\geqslant 1-R^s$ and \emph{bottom-tier} if $v<1-R^s$. From the definition of $R^s$ it follows that $(1-\under u)v \leqslant R^s$ for every bottom-tier project. 

Under $\rho^\ast$, the agent's expected payoff is $\underu$ if all of the proposed projects are bottom-tier, and is at least $\underu $ if at least one of the proposed projects is top-tier. Therefore, if the agent has some top-tier project, it is optimal to propose the agent's favorite top-tier project $(u,v)$, in which case the regret is at most $1-v\leqslant R^s$. If all the projects in $A$ are bottom-tier, it is optimal to propose a project $(u',v')$ that maximizes $\alpha^s(u,v)v$. Let $(u_p,v_p)\in \argmax_{(u,v)\in A} v$ be a principal's favorite project. Thus, the regret is $v_p- \alpha^s(u',v')v' \leqslant v_p-\alpha^s(u_p,v_p)v_p \leqslant (1-\underu)v_p\leqslant R^s$.  
      
     We prove that this choice function is admissible in part (iii).  
     
    \vspace{-0.35cm} \item If $P$ includes a top-tier project or $u_p(P)=0$, then $\rho((\cdot,v_p(P))|P) \leqslant 1=\rho^\ast((\cdot,v_p(P))|P)$. If all projects in $P$ are bottom-tier and $u_p(P)>0$, we claim that $\rho((\cdot,v_p(P))|P)\leqslant  \underu /u_p(P)=\rho^\ast((\cdot,v_p(P))|P)$. Suppose otherwise that $\rho((\cdot,v_p(P))|P)> \underu /u_p(P)$ for some $P$. If the agent's type is $P \cup \{(\underu, 1)\}$, then he will not include $(\underu, 1)$ in his proposal, since proposing just $P$ leads to a strictly higher payoff than $\underu$. As a result, the regret is strictly higher than $R^s$, contradicting the fact that the $\wcr$ under $\rho$ is at most $R^s$.  

% Since $\rho$ has the $\wcr$ of $R^s$, 
Lastly, we show that the choice function in statement (ii) is admissible. Fix a mechanism $\rho$ which has the $\wcr$ of $R^s$. If the agent proposes only bottom-tier projects, his expected payoff cannot exceed $\under u$. (Otherwise, the agent has incentive to hide $(\underu,1)$, causing a higher regret than $R^s$.)

Consider the set of all single-project proposals with one top-tier project. If there exists one such proposal $P$ such that $\rho$ does not choose the proposed project with probability $1$, then if the agent's type is actually $P$, $\rho$ does strictly worse than $\rho^*$ for the principal. 

Now suppose that, for every single-project proposal with one top-tier project, $\rho$ chooses the proposed project with probability $1$. Then for every $A$ with at least one top-tier project, $\rho$ and $\rho^\ast$ implement the same project choice, which is the agent's favorite top-tier project in $A$. Hence, we only need to consider those $A$ with only bottom-tier projects. Since $\rho((\cdot,v_p(P))|P) \leqslant \rho^\ast((\cdot,v_p(P))|P)$, the principal's payoff is higher under $\rho^\ast$ than under $\rho$ for every $P \subset A$. Moreover, under the choice function in (ii), the agent chooses a project in $A$ to maximize the principal's expected payoff, so $\rho$ cannot do better than the choice function in (ii) for every $A$ with only bottom-tier projects.  
   \end{enumerate}
	
\end{proof}

     \nocite{*}
\bibliographystyle{aea}
\bibliography{cite}

\end{document}